\documentclass[journal]{IEEEtran}
\IEEEoverridecommandlockouts
\usepackage{makecell}
\usepackage{cite}
\usepackage{amsmath,amssymb,amsfonts}
\usepackage{algorithm}
\usepackage{array}
\usepackage{stfloats}
\usepackage{url}
\usepackage{verbatim}

\usepackage{graphicx}

\usepackage{upgreek}
\usepackage{graphicx,color,overpic,psfrag}
\usepackage{textcomp}
\usepackage[T1]{fontenc}
\usepackage{mathrsfs}
\usepackage{float}
\graphicspath{{Figures/}}
\usepackage[table]{xcolor}
\usepackage[bookmarks=true, colorlinks, linkcolor=black, anchorcolor=black, citecolor=black]{hyperref}
\usepackage{latexsym}
\usepackage{bm}
\usepackage{cases}
\usepackage{fancyhdr}
\usepackage{setspace}
\usepackage{algpseudocode}
\usepackage{blkarray}
\usepackage{booktabs}
\usepackage{multirow}
\usepackage{dsfont}
\usepackage{tabularx}
\usepackage{amsfonts}
\usepackage{tikz}
\usepackage{lipsum}
\graphicspath{{Figures/}}
\newtheorem{theorem}{Theorem}

\newtheorem{remark}{Remark}

	\DeclareMathAlphabet\mathbfcal{OMS}{cmsy}{b}{n}
	\ifCLASSOPTIONcompsoc
	\usepackage[caption=false,font=normalsize,labelfont=sf,textfont=sf]{subfig}
	\else
	\usepackage[caption=false,font=footnotesize]{subfig}
	\fi
	
	\allowdisplaybreaks[4]

	\newcommand{\ra}[1]{\renewcommand{\arraystretch}{#1}} 
	
	

	\def\BibTeX{{\rm B\kern-.05em{\sc i\kern-.025em b}\kern-.08em
			T\kern-.1667em\lower.7ex\hbox{E}\kern-.125emX}}
	
\begin{document}
		\title{Channel Charting With Physical Channel Fingerprints For Massive MIMO-OFDM Channel Acquisition}

		\author{Jinke Tang,~\IEEEmembership{Graduate~Student~Member, IEEE,} Xiqi~Gao,~\IEEEmembership{Fellow, IEEE},  Li You,~\IEEEmembership{Senior~Member, IEEE,}  Xiang-Gen~Xia,~\IEEEmembership{Fellow, IEEE,} and Cheng-Xiang Wang,~\IEEEmembership{Fellow, IEEE}

			
			\thanks{Jinke Tang, Xiqi Gao,  Li You, and Cheng-Xiang Wang are with the National Mobile Communications Research Laboratory, Southeast University, Nanjing 210096, China, and also with the Purple Mountain Laboratories, Nanjing 211111, China  (e-mail: jktang@seu.edu.cn,
				xqgao@seu.edu.cn, lyou@seu.edu.cn, chxwang@seu.edu.cn).}
			\thanks{Xiang-Gen Xia is with the Department of Electrical and Computer Engineering, University of Delaware, Newark, DE 19716, USA (e-mail: xianggen@udel.edu).}
		}

		\maketitle
		\begin{abstract}
			
			The advancement of 6G mobile communication and positioning technologies has amplified the significance of location-aware tools, such as location-indexed channel fingerprints (CFs) and channel charting, which are becoming key enablers for massive MIMO-OFDM systems. In this paper, we propose a novel channel charting with physical CFs (PCFs) and demonstrate its effectiveness in channel state information (CSI) acquisition. First, we define the PCF based on a cluster-based geometric stochastic channel model (GBSM), enabling a comprehensive representation of physical channel characteristics using a compact set of parameters. We then develop a methodology for PCF acquisition in massive MIMO-OFDM systems. By exploiting the relationship between PCFs and the space-frequency-time (SFT) domain channel, the proposed method extracts PCFs from multi-location channel measurements and constructs a structured channel charting with location-indexed PCFs. Furthermore, we propose a low-complexity algorithm to acquire beam domain statistical CSI (sCSI) using the PCFs in the channel charting. The resulting sCSI can be directly employed as prior information for channel estimation. Simulation results show that the proposed method delivers sCSI performance comparable to traditional online probing techniques, and the generated sCSI can serve as reliable prior knowledge to significantly enhance the accuracy of channel estimation. These results validate the proposed PCF as a powerful and versatile tool for channel acquisition and system design of the next-generation mobile communication.

		\end{abstract}
		
		\begin{IEEEkeywords}
			Massive MIMO-OFDM, channel fingerprint, sCSI, cluster-based GBSM, beam domain.
		\end{IEEEkeywords}
		\section{Introduction}\label{section1}

		Research on sixth-generation (6G) mobile communication technology is advancing rapidly and has attracted substantial global attention \cite{you2021towards}. Compared to 5G systems, 6G is expected to fulfill more demanding performance requirements, including significantly higher user terminal (UT) connectivity density, enhanced data rates, reduced latency \cite{you2021towards,9237116,10158439}, and improved positioning accuracy \cite{10286475}. Given its proven effectiveness in 5G networks, massive multi-input multi-output orthogonal frequency division multiplexing (MIMO-OFDM) is expected to remain a crucial technology in 6G. The adoption of massive MIMO arrays and their evolution into ultra-massive MIMO will be pivotal in addressing the rigorous performance demands of 6G \cite{you2021towards}.

		However, the increasing density of UTs and the deployment of large-scale antenna arrays in 6G networks introduce significant challenges for system design, particularly in the acquisition of channel state information (CSI) \cite{mine}. CSI is typically categorized into instantaneous CSI and statistical CSI (sCSI). Unlike instantaneous CSI, sCSI evolves relatively slow over time and can serve as valuable prior knowledge to assist in the estimation of instantaneous CSI \cite{9910031, mine, 10026502}. In scenarios involving a large number of UTs, sCSI is especially critical for guiding pilot scheduling, thereby effectively mitigating pilot interference among UTs \cite{mine}. Additionally, sCSI plays a key role in key tasks such as precoding, UT grouping, and power control \cite{10146318, 9003618,9592715}. Most existing studies obtain sCSI through online probing methods based on pilot signals with various schemes proposed in \cite{8074806,10146318,9003618,9592715, 10993474}. However, as both the number of UTs and the dimensionality of antenna arrays continue to increase, these probing-based methods might be impractical due to their excessive overhead and computational complexity.

		Meanwhile, the evolution of mobile communication systems also presents new opportunities for sCSI acquisition. Massive MIMO-OFDM systems facilitate the collection of high-resolution channel parameters, including path gain, angle, and delay \cite{sage11}. Furthermore, the dense deployment of UT nodes in future systems enables finer spatial sampling granularity \cite{you2021towards}. By treating base stations (BSs) and UTs as distributed sensors, a high-precision channel parameter database can be constructed \cite{9373011, 10430216}. This database can be employed to construct a channel charting \cite{8444621,10495336,11223663,9968109,che2024channel} or a channel knowledge map (CKM) \cite{10430216} to capture environment-specific propagation characteristics across the entire cell. Specifically, given the location of a UT, the corresponding sCSI can be inferred from the database, eliminating the need for online probing and enabling training-free channel acquisition \cite{10430216}.

		In this paper, we refer to this structured, location-indexed
		representation as channel charting with channel fingerprints
		(CFs). Here, “CFs” denote location-dependent channel parameters, while “charting” highlights the spatially organized structure. Existing channel charting techniques \cite{8444621,10495336,11223663,9968109,che2024channel} typically rely on unsupervised learning (UL) to embed high-dimensional channel measurements into a latent space. Although these data-driven embeddings preserve spatial proximity, they do not explicitly capture the underlying physical channel properties. As a result, they are not directly suitable for reconstructing sCSI or for tasks that require physically meaningful parameters. In contrast, our proposed approach extracts CFs directly from the underlying physical channel model, ensuring that each CF has clear physical meaning and faithfully captures the intrinsic location- and environment-dependent characteristics of the wireless channel. Because these CFs are physics-driven rather than data-driven latent embeddings, they provide an intuitive and accurate geometric representation of the spatial channel variations. Moreover, their physically meaningful structure naturally forms a reliable location-indexed database that can be directly used for downstream tasks such as sCSI acquisition.

		The idea of utilizing CFs in channel charting to replace online probing for sCSI acquisition raises a fundamental question: which CFs are suitable for this purpose? This application scenario imposes specific requirements: an ideal CF should be general-purpose, storage-efficient, and adaptable to varying needs. Existing CFs, however, are often designed with different objectives in mind. Some, like angle-based or beam-index CFs for beam alignment \cite{10287775,10791446} or spatial covariance matrices for precoding  \cite{xie}, are application-specific and capture only partial channel characteristics, limiting their suitability for general sCSI reconstruction. Others, such as the spatial-frequency (SF) beam domain channel power distribution matrix, which is a kind of comprehensive sCSI for channel estimation \cite{9910031} and precoding \cite{9592715}, while rich in information, may lead to considerable storage overhead in large-scale deployments \cite{10292876,11247875}. Moreover, these CFs are typically tied to a fixed beam domain configuration, reducing their applicability when different sCSI types or beam domain dimensions are required.    To address these challenges, we propose a novel channel charting with physical channel fingerprints (PCFs) in this paper, which extends existing channel charting and CKM representations in a unified and physics-driven manner. Specifically, the proposed PCF exhibits the following key properties:
		\begin{itemize}
			\item \emph{Generality $\&$ Completeness}: By encoding a complete set for channel properties, the PCF at a given location can be transformed into specific kinds of sCSI to meet various application requirements.
			\item \emph{Storage Efficiency}:  The PCF compresses channel properties into a compact parameter set, drastically reducing storage overhead while preserving reconstruction capability.
			\item \emph{Physics-Driven}: Derived from the propagation environment, the PCF is not tied to any predefined sCSI type or beam domain definition, allowing a single stored PCF to be adaptively projected into a desired sCSI format.
		\end{itemize}

		To distinguish the proposed PCF from these prior channel charting (or CKM) approaches, we provide a comparative summary in Table I. The key innovation of PCF is that, unlike prior CFs or latent coordinates, it represents each location using a physics-driven parameter set that can be analytically transformed into any domain of sCSI. This unified design simultaneously achieves physical interpretability, functional versatility, and storage scalability, making PCF a suitable and generalizable tool for sCSI acquisition.

		\definecolor{lightblue}{rgb}{0.93,0.95,1.0}
		\begin{table*}[t]
			\captionsetup{font=footnotesize}
			\caption{Characteristics and Applicability of Channel Charting Approaches}
			\label{CFcompare}
			\centering
			\scriptsize
			\renewcommand{\arraystretch}{1.5}
			\setlength{\tabcolsep}{2pt} 
			\begin{tabular}{>{\raggedright\arraybackslash}p{3.2cm}
					>{\raggedright\arraybackslash}p{3.5cm}
					>{\raggedright\arraybackslash}p{5cm}
					>{\raggedright\arraybackslash}p{4.4cm}}
				\toprule
				\textbf{Aspect} &
				\textbf{UL-based Channel Charting} &
				\textbf{Channel Charting with Conventional CFs} &
				\textbf{Channel Charting with PCFs (Proposed)} \\
				\midrule
				\rowcolor{lightblue}
				Modeling Basis &
				Data-driven (learned latent space) &
				Physics-driven (some derived from partial physical properties [21], [22])  &
				\textbf{Physics-driven(considering complete physical properties)} \\
				Output Format & 
				Low-dimensional latent embedding &
				Task-specific channel features &
				\textbf{Compact parameter set} \\
				\rowcolor{lightblue}
				Physical Interpretation &
				Not directly available &
				Interpretable within the target task or system design &
				\textbf{Direct (complete set summarizing physical properties)} \\
				Primary Application &
				Similarity-based tasks (e.g., UT grouping [16, 17], PSOP reuse [19, 20]) &
				Dedicated functions (e.g., beam alignment [21, 22], precoding with fixed sCSI [23-25]) &
				\textbf{General tasks (sCSI acquisition, channel estimation, precoding)} \\
				\rowcolor{lightblue}
				Storage Overhead &
				Low (compact embedding) &
				Varies (can be high for full sCSI storage [24]) &
				\textbf{Low (compact parameter representation)} \\
				Adaptability to New Tasks/Configurations &
				Limited (requires retraining) &
				Limited (features tied to original design) &
				\textbf{High (parameters projectable to new beam domains/sCSI types)} \\
				\bottomrule
			\end{tabular}
		\end{table*}

		Due to its physics-driven properties, the generation of PCFs in massive MIMO communications relies on appropriate channel modeling techniques. Commonly adopted models include geometric-based stochastic models (GBSMs) \cite{9786750, 9318511, huang2020geometry, 10971994} and beam-based channel models (BBCMs) \cite{9954203, 9910031, mine, 10026502, 7332961}. Among them, GBSMs explicitly characterize the properties of multipath components (MPCs) in the propagation environment, making them well-suited for analyzing the physical properties of wireless channels. Therefore, we define the PCF within the GBSM framework. In contrast, BBCMs, as a type of statistical channel models, quantize the continuous angular and delay domains into discrete beams, thereby reducing the complexity of channel processing. This property has led to their widespread use in channel acquisition and transmission design \cite{9910031, mine, 10026502, 7332961}. To bridge these two models, we model the same propagation environment using both GBSM and BBCM. Based on GBSM, we define the PCF, and by leveraging the intrinsic relationship between GBSM and BBCM, we further propose a method to transform the PCF into beam domain sCSI for further applications.

		In summary, this paper focuses on the channel charting with PCFs, systematically addressing its formal definition, acquisition methodology, and applications in sCSI acquisition. The contributions are summarized as follows.
		\begin{itemize}
			\item We introduce the definition of PCF based on a cluster-based GBSM. According to the clustering structure of MPCs in the GBSM, the PCF is defined to capture the physical characteristics of the channel comprehensively using a compact set of parameters.
			\item We propose a methodology for acquiring PCFs in massive MIMO-OFDM systems. This includes formulating the SFT domain channel model and establishing its connection to physical channel parameters, estimating the SFT domain channels, and employing a space-alternating generalized expectation-maximization (SAGE)-based algorithm to extract PCFs from the estimated channels. By aggregating these location-specific PCFs, a channel charting with PCFs can be constructed, providing a location-indexed representation of the physical channel environment.
			\item We propose an efficient method for sCSI acquisition using the PCFs stored in the channel charting. We firstly establish the relationship between the cluster-based GBSM and the triple BBCM. Leveraging this relationship, we derive the expression of beam domain sCSI and develop a low-complexity algorithm for its generation from the PCF. The resulting sCSI can be directly utilized as prior information to achieve efficient channel estimation.
		\end{itemize}

		The rest of the paper is organized as follows. In Section~\ref{section2}, we introduce the definition of PCF based on a GBSM. In Section~\ref{section3}, we propose the methodology for PCF acquisition and construct the channel charting with PCFs. In Section~\ref{section4}, the method for sCSI acquisition with PCFs is proposed. Simulation results are shown in Section~\ref{section5}, and the paper is concluded in Section~\ref{section6}.

		\emph{Notations:} In this paper, bold lowercase (uppercase, calligraphic) letters are used to denote column vectors (matrices, tensors). ${\bf{0}}$ denotes the all-zeros vector, matrix or tensor. ${{\bf{I}}_N}$ denotes the $N \times N$ identity matrix. Superscript ${( \cdot )^{\rm T}}$, ${( \cdot )^{\rm H}}$, ${( \cdot )^*}$, ${( \cdot )^{-1}}$ denote the transpose, conjugate-transpose, conjugate, inverse operations, respectively. We adopt ${\left[ {\bf{a}} \right]_i}$, ${\left[ {\bf{A}} \right]_{i,j}}$ and ${\left[ {\mathbfcal A} \right]_{i,j,k}}$ to denote the $i$-th element of the vector ${\bf{a}}$, the $(i,j)$-th element of ${\bf{A}}$ and the $(i,j,k)$-th element of the tensor ${\mathbfcal A}$, respectively. All the indexes in this paper are started from $1$. $\otimes$ and $\odot $ denote the Kronecker product and the Hadamard product, respectively. The operator ${\rm Diag}\{  \cdot \} $ denotes the diagonal matrix of the vector along its main diagonal, while ${\rm diag}\{  \cdot \} $ denotes extracting the diagonal elements as a vector. ${\rm tr}\{  \cdot \}$ denotes the matrix trace operation and ${\rm{vec}}\left\{  \cdot  \right\}$ denotes the vectorization operation. 
		${{\mathbb{C}}^{M \times N}}$ and ${{\mathbb{R}}^{M \times N}}$ denote the ${M \times N}$ dimensional complex and real vector spaces, respectively. ${\mathbb{E}}\{  \cdot \} $ denotes the expectation operation. We adopt $ {{{\cal A} _1}} \times  {{{\cal A} _2}}$ as the Cartesian product of the sets ${{\cal A} _1}$ and ${{\cal A} _2}$. $\delta \left(  \cdot  \right)$ denotes the Dirac delta function. The set $\{ {a_p}\} _{p = 1}^P$ denotes $ \{ {a_1},{a_2}, \cdots ,{a_P}\} $. ${{\bf{F}}_N}$ denotes the $N\times N$ discrete Fourier transform (DFT) matrix. ${\bf{\Gamma }}_{M,m}$ represents
		\begin{equation}\label{Gamma}
			{\bf{\Gamma }}_{M,m} \buildrel \Delta \over = \left[ {\begin{array}{*{20}{c}}
					{\bf{0}}&{{{\bf{I}}_{M - {{\left\langle m \right\rangle }_M}}}}\\
					{{{\bf{I}}_{{{\left\langle m \right\rangle }_M}}}}&{\bf{0}}
			\end{array}} \right].
		\end{equation}
		
		Given a tensor  ${\mathbfcal X} \in {\mathbb{C}}^{{I_1} \times {I_2} \times  \cdots  \times {I_m} \times  \cdots  \times {I_M}}$ and a matrix ${\bf{X}} \in {{\mathbb{C}}^{{J} \times {I_m}}}$ , the \emph{m-mode product} between ${\mathbfcal X}$ and ${\bf{X}}$ is denoted as ${\bf{X}}{ \times _m}{\mathbfcal X}$ $ \in$ $ {{\mathbb{C}}^{{I_1} \times  \cdots  \times {I_{m - 1}} \times {J} \times {I_{m + 1}} \times  \cdots  \times {I_M}}}$. Elementwise, we have    \vspace{-0.1cm}
		\begin{equation}\label{mmode}
			\!\!\!{\left[ {{\bf{X}}{ \times _m}{\mathbfcal X}} \right]_{{i_1}, \cdots ,{j}, \cdots ,{i_M}}} = \sum\limits_{{i_m} = 1}^{{I_m}} {{{\left[ {\bf{X}} \right]}_{{j},{i_m}}}{{\left[ {\mathbfcal X} \right]}_{{i_1}, \cdots ,{i_m}, \cdots ,{i_M}}}}.
		\end{equation}
		
		With $\!\!{\mathbfcal A}$$\quad\!\!\!\! \in$ $\quad\!\!\!\! {{\mathbb{C}}^{{I_1} \times {I_2} \times  \cdots  \times {I_M} \times {J_1} \times {J_2} \times  \cdots  \times {J_K}}}$ and ${\mathbfcal B}$ $\quad\!\!\!\! \in $ $\quad\!\!\!\!{{\mathbb{C}}^{{J_1} \times {J_2} \times  \cdots  \times {J_K} \times {P_1} \times {P_2} \times  \cdots  \times {P_N}}}$, the \emph{Einstein} \emph{product} between ${\mathbfcal A}$ and ${\mathbfcal B}$ is denoted as ${{\mathbfcal A}{*_{K}}{\mathbfcal B}}$$ \in {{\mathbb{C}}^{{I_1} \times  \cdots  \times {I_M} \times {P_1} \times  \cdots  \times {P_N}}}$. Elementwise, we have 
		\begin{equation}\label{einstein}\resizebox{0.85\hsize}{!}
			{$\begin{array}{l}
					\!\!\!\!\!\!	{\left[ {{\mathbfcal A}{*_{K}}{\mathbfcal B}} \right]_{{i_1}, \cdots ,{i_M},{p_1}, \cdots ,{p_N}}}=\\\!\!\!\!\displaystyle\sum\limits_{{j_1} = 1}^{{J_1}} { \cdots \displaystyle\sum\limits_{{j_K} = 1}^{{J_K}} {{{\left[ {\mathbfcal A} \right]}_{{i_1}, \cdots ,{i_M},{j_1}, \cdots ,{j_K}}}} } {{\left[ {\mathbfcal B} \right]}_{{j_1}, \cdots ,{j_K},{p_1}, \cdots ,{p_N}}} .
				\end{array}$}
		\end{equation}
		

		

		\section{Physical Channel Fingerprints}\label{section2}

		In this section, we introduce the definition of PCF based on a cluster-based GBSM, which compactly and comprehensively capture the physical characteristics of the wireless channel using a set of parameters.
		
		\subsection{Cluster-based GBSM}
		
		We consider the propagation environment as shown in Fig.~\ref{txtorx1}, where the BS communicates with several UTs on the ground. Using UT$_1$ in Fig.~\ref{txtorx1} as an example, we introduce the channel modeling method between the BS and the UT. For analytical tractability, assume both the BS and UT employ uniform linear arrays (ULAs), enabling an effective two-dimensional (2D) propagation characterization. Additionally, as Fig.~\ref{txtorx2} shows, the scenario is discretized into multiple grids, assuming that the propagation environment within each grid remains approximately unchanged. For notational clarity, we define ${\bf{s}}_{\mathrm{BS}}$ as the BS grid coordinate and $\mathcal{S}$ as the set of all potential UT grid locations within the coverage area.


		\begin{figure}[t!]
			\centering
			\includegraphics[width =130pt]{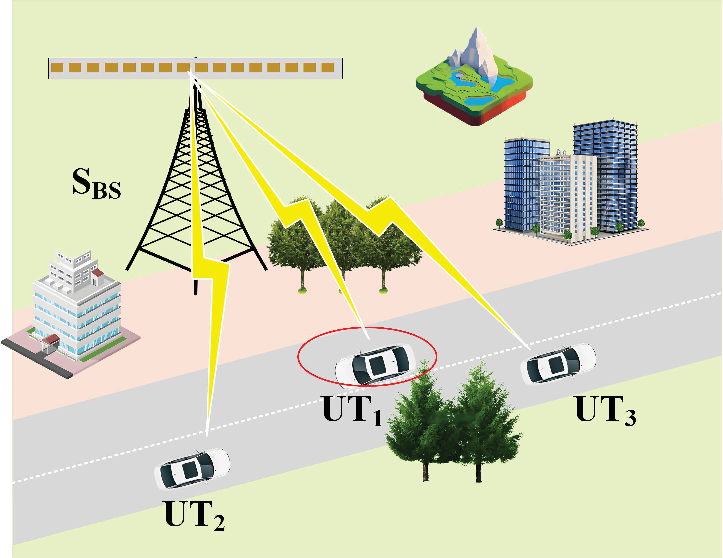}
			\caption{The propagation environment between a BS and several mobile UTs.}
			\label{txtorx1}
			
		\end{figure}

		We analyze the channel impulse response (CIR) between the UT and the BS over a given time period. During this period, the UT movement remains confined within a specific grid ${\bf{s}}$, moving at a speed of ${v_{{\rm{move}}}}$ in the direction ${{\alpha}_{{\rm{move}}}}$. Consequently, the propagation environment between the UT and the BS can be considered approximately static. The CIR function between the $a$-th antenna element at the BS and the $b$-th antenna element at the UT can therefore be expressed as \cite{9786750,9318511}
		\begin{equation}\label{nocluster}
			{{\tilde h}_{a,b}}\left( {t,\tau } \right) = \sum\limits_{l = 1}^{L({\bf{s}})} {\kappa _{{\bf{s}},l}^{}{e^{ - \bar \jmath 2\pi {f_{\rm{c}}}\tau _{{\bf{s}},a,b,l}^{}(t)}}\delta \left( {\tau  - \tau _{{\bf{s}},a,b,l}^{}(t)} \right)},	
		\end{equation}
		where ${L({{\bf{s}}})}$  is the total number of MPCs between the BS and ${{\bf{s}}}$,  ${\kappa _{{\bf{s}},l}^{}}$ and ${\tau _{{\bf{s}},a,b,l}(t)}$ denote the complex-valued gain and the delay of the $l$-th MPC, respectively. ${{f_{\rm{c}}}}$ represents the carrier frequency.

		For the uplink, the angle of departure (AoD) and the angle of arrival (AoA) of the $l$-th MPC are represented by ${\beta _{{\bf{s}},l}}$ and ${\theta _{{\bf{s}},l}}$, respectively. Then under the far-field assumption, the delay can be approximated as \cite{9910031}
		\begin{align}\label{tauab}
			\tau _{{\bf{s}},a,b,l}^{}(t) \approx& \tau _{{\bf{s}},l}^{} + \scalebox{1.2}{$\frac{{(a - 1){d_{{\rm{BS}}}}\cos (\theta _{{\bf{s}},l}^{})}}{c} + \frac{{(b-1){d_{{\rm{UT}}}}\cos (\beta _{{\bf{s}},l}^{})}}{c}$}\notag\\
			& \scalebox{1.2}{$- \frac{{{v_{{\rm{move}}}}t\cos \left( {{\alpha _{{\rm{move}}}} - \beta _{{\bf{s}},l}^{}} \right)}}{c}$},
		\end{align}
		where ${{d_{{\rm{BS}}}}}$ and ${{d_{{\rm{UT}}}}}$ represent the inter-antenna spacing at the BS and UT sides, respectively, and $\tau _{{\bf{s}},l}$ is determined by the traveling distance of the $l$-th MPC between the coordinates ${{{\bf{s}}_{{\rm{BS}}}}}$ and ${{{\bf{s}}}}$. $c$ represents the speed of the light. $\tau _{{\bf{s}},l}$ represents the propagation time of the $l$-th MPC between the coordinates ${{{\bf{s}}{\rm BS}}}$ and ${{\bf{s}}}$.

		By substituting (\ref{tauab}) into (\ref{nocluster}), the CIR can be expressed as
		\begin{align}\label{cir1}
			\!\!\!\!\!\!\!&\scalebox{0.94}{$\!\!\!\!\!\!\!{{\tilde h}_{a,b}}\left( {t,\tau } \right) = \sum\limits_{l = 1}^{L({\bf{s}})} {\underbrace {\kappa _{{\bf{s}},l}^{}{e^{ - \bar \jmath 2\pi {f_{\rm{c}}}\tau _{{\bf{s}},l}^{}}}}_{\alpha _{{\bf{s}},l}^{}}{e^{\bar \jmath 2\pi \nu _{{\bf{s}},l}^{}t}}} $} \notag\\
			&\scalebox{0.94}{${e^{ - \bar \jmath 2\pi {f_{\rm{c}}}\left( {\frac{{(a-1){d_{{\rm{BS}}}}\cos (\theta _{{\bf{s}},l}^{})}}{c} + \frac{{(b-1){d_{{\rm{UT}}}}\cos (\beta _{{\bf{s}},l}^{})}}{c}} \right)}}\delta \left( {\tau  - \tau _{{\bf{s}},a,b,l}^{}(t)} \right)$},
		\end{align}
		where ${\nu _{{\bf{s}},l}} = {v_{{\rm{move}}}}\cos \left( {{\alpha _{{\rm{move}}}} - {\beta _{{\bf{s}},l}^{}}} \right)/{\lambda_c}$ represents the Doppler shift of the $l$-th MPC between ${{{\bf{s}}_{{\rm{BS}}}}}$ and ${{\bf{s}}}$, and $\lambda _c$ represents the wavelength.
		
		According to  \cite{huang2020geometry, 9497668, 8693871}, MPCs with similar propagation characteristics can typically be grouped into a cluster, as illustrated in Fig.~\ref{txtorx2}. Therefore, the CIR in (\ref{cir1}) can be further expressed as
		\begin{align}\label{cir2}
			&\scalebox{0.9}{${{\tilde h}_{a,b}}\left( {t,\tau } \right) = \sum\limits_{p = 1}^{P({\bf{s}})} {\sum\limits_{q = 1}^{{L_p}({\bf{s}})} {\alpha _{{\bf{s}},l{}_{p,q}}^{}{e^{\bar \jmath 2\pi \nu _{{\bf{s}},{l_{p,q}}}^{}t}}} } $} \notag\\
			&\scalebox{0.89}{${e^{ - \bar \jmath 2\pi {f_{\rm{c}}}\left( {\frac{{(a-1){d_{{\rm{BS}}}}\cos \left( {\theta _{{\bf{s}},{l_{p,q}}}^{}} \right)}}{c} + \frac{{(b-1){d_{{\rm{UT}}}}\cos \left( {\beta _{{\bf{s}},{l_{p,q}}}^{}} \right)}}{c}} \right)}}\!\!\delta \left( {\tau  - \tau _{{\bf{s}},a,b,{l_{p,q}}}^{}(t)} \right)$},
		\end{align}
		where ${\alpha _{{\bf{s}},l_{p,q}}}$, ${\theta _{{\bf{s}},{l_{p,q}}}}$, ${\beta _{{\bf{s}},{l_{p,q}}}^{}}$, ${\tau _{{\bf{s}},a,b,{l_{p,q}}}(t)}$ and ${\nu _{{\bf{s}},{l_{p,q}}}}$ denote the complex-valued gain, AoA,  AoD, delay and Doppler shift of the $q$-th MPC in the $p$-th cluster between ${\bf{s}}$ and ${{{\bf{s}}_{{\rm{BS}}}}}$, respectively. ${{l_{p,q}}}$ denotes the index of the $q$-th MPC in the $p$-th cluster among all MPCs. ${P({{\bf{s}}})}$ is the number of clusters between ${{{\bf{s}}}}$ and ${{{\bf{s}}_{\rm {BS}}}}$, and ${L_p({{\bf{s}}})}$ is the number of MPCs in the $p$-th cluster. We have $L({\bf{s}}) = \sum\nolimits_{p = 1}^{P({\bf{s}})} {{L_p}({\bf{s}})}$. Within the same cluster, the angle and delay distributions of different MPCs can typically be well approximated by specific probability distribution functions (PDFs) \cite{huang2020geometry, 9497668, 8693871, 10225614, 9786750, 9318511}. Moreover, since MPCs within a cluster are generally spatially close to each other, as shown in  Fig.~\ref{txtorx2}, for $q \ne q'$, ${\mathbb E}\{ |{\alpha _{{\bf{s}},{l_{p,q}}}}{|^2}\}  \approx {\mathbb E}\{ |{\alpha _{{\bf{s}},{l_{p,q'}}}}{|^2}\}  $ holds \cite{9786750, 9318511}. Compared with (\ref{cir1}), the model in (\ref{cir2}) reflects the clustering of MPCs in the  environment, it is commonly referred to as a cluster-based GBSM.
		%
		

		\begin{figure}[t!]
			\centering
			\includegraphics[width =220pt]{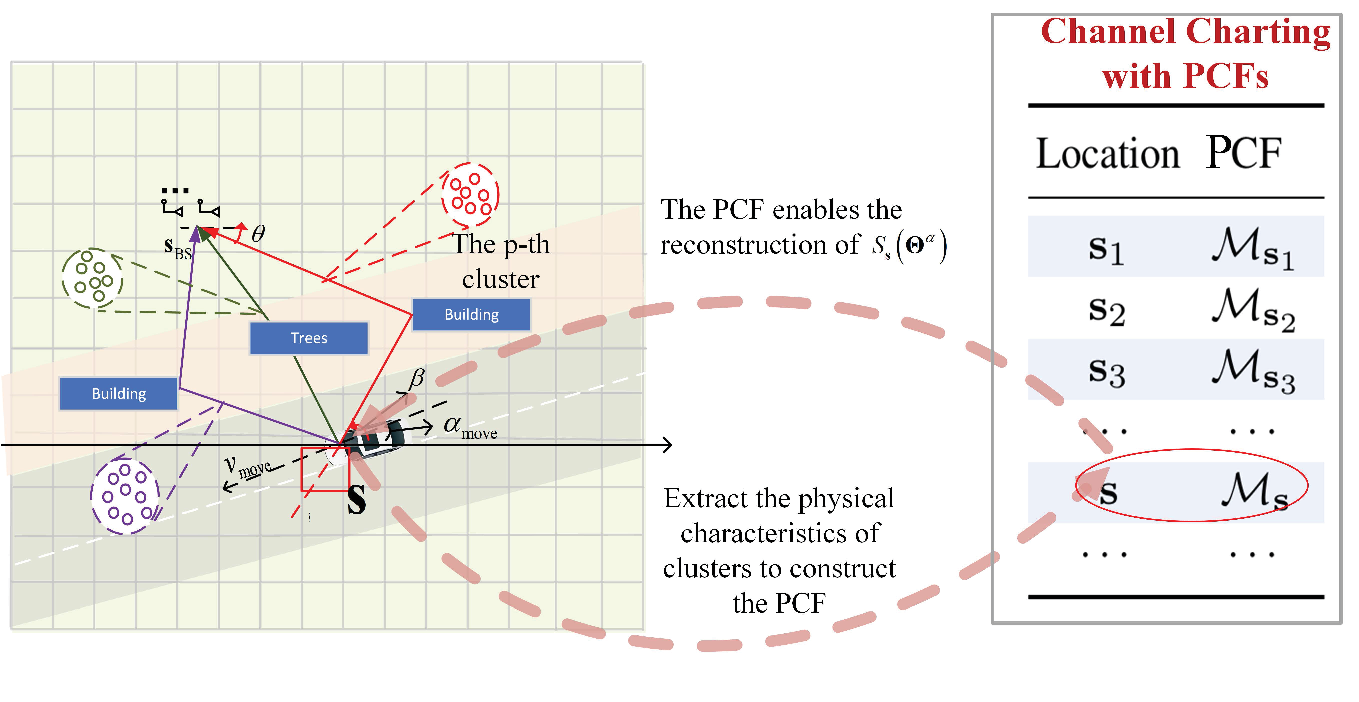}
			\caption{The relationship between the GBSM and PCFs.}
			\label{txtorx2}
			
		\end{figure}
		
		
		We now consider a single-cell TDD MIMO-OFDM system, where ${N_{\rm c}}$ denotes the number of subcarriers and ${N_{\rm g}}$ is the length of the cyclic prefix (CP). Given a sampling interval length of $T_{\rm s}$, the subcarrier spacing is defined as $\Delta f = \frac{1}{{{N_{\rm c}}{T_{\rm s}}}}$, and the duration of each OFDM symbol is ${T_{{\rm{sym}}}} = ({N_{\rm c}} + {N_{\rm g}}){T_{\rm s}}$. We assume that the channel remains quasi-static within each OFDM symbol but varies across symbols due to the Doppler effect. Under this assumption, the CIR in (\ref{cir2}) can be transformed into the SFT domain. The corresponding SFT channel transfer function (CTF) for the $k$-th subcarrier in the $n$-th OFDM symbol is given by:
		\begin{align}\label{CTF11}
			&\scalebox{0.9}{${h_{a,b}}\left( {n{T_{{\rm{sym}}}},k\Delta f} \right) = \sum\limits_{p = 1}^{P({\bf{s}})} {\sum\limits_{q = 1}^{{L_p}({\bf{s}})} {\alpha _{{\bf{s}},{l_{p,q}}}^{}{e^{\bar \jmath 2\pi \nu _{{\bf{s}},{l_{p,q}}}^{}n{T_{{\rm{sym}}}}}}} }$} \notag\\
			&\scalebox{0.85}{${e^{ - \bar \jmath 2\pi {f_{\rm{c}}}\left( {\frac{{(a-1){d_{{\rm{BS}}}}\cos \left( {\theta _{{\bf{s}},{l_{p,q}}}^{}} \right)}}{c} + \frac{{(b-1){d_{{\rm{UT}}}}\cos \left( {\beta _{{\bf{s}},{l_{p,q}}}^{}} \right)}}{c}} \right)}}{e^{-\bar \jmath 2\pi k\Delta f{\tau _{{\bf{s}},a,b,{l_{p,q}}}}(n{T_{{\rm{sym}}}})}}.$}	     
		\end{align}
		For $n$ satisfying ${v_{{\rm{move}}}}n{T_{\rm{sym}}}/c \ll {T_s}$, (\ref{CTF11}) can be further written as \cite{Shen}
		\begin{align}\label{CTF1}
			&\scalebox{0.9}{${h_{a,b}}\left( {n{T_{{\rm{sym}}}},k\Delta f} \right) = \sum\limits_{p = 1}^{P({\bf{s}})} {\sum\limits_{q = 1}^{{L_p}({\bf{s}})} {\alpha _{{\bf{s}},{l_{p,q}}}^{}{e^{\bar \jmath 2\pi \nu _{{\bf{s}},{l_{p,q}}}^{}n{T_{{\rm{sym}}}}}}} {e^{-\bar \jmath 2\pi k\Delta f{\tau _{{\bf{s}},{l_{p,q}}}}}}} $}\notag\\
			&\scalebox{0.9}{${e^{ - \bar \jmath 2\pi \left( {{f_{\rm{c}}} + k\Delta f} \right)\left( {\frac{{(a-1){d_{{\rm{BS}}}}\cos \left( {\theta _{{\bf{s}},{l_{p,q}}}^{}} \right)}}{c} + \frac{{(b-1){d_{{\rm{UT}}}}\cos \left( {\beta _{{\bf{s}},{l_{p,q}}}^{}} \right)}}{c}} \right)}}.$}
		\end{align}

		For any ${\bf{s}} \in {\cal S}$, given the parameters $[{\alpha _{{\bf{s}},{l_{p,q}}}^{}},{{\beta _{{\bf{s}},{l_{p,q}}}}},$ ${{\theta _{{\bf{s}},{l_{p,q}}}}},{{\tau _{{\bf{s}},{l_{p,q}}}}},$ ${\nu _{{\bf{s}},{l_{p,q}}}^{}}$]$^{\rm T}$ for each MPC, the SFT domain channel at the corresponding location can be reconstructed according to (\ref{CTF1}). Furthermore, if the UT mobility is known, the Doppler shift ${\nu _{{\bf{s}},{l_{p,q}}}^{}}$ can be calculated from ${{\beta _{{\bf{s}},{l_{p,q}}}}}$. By defining the channel gains as a function of ${{\bf{\Theta }}^\alpha } = [\beta, \theta, \tau]^{\rm T}$ as
		\begin{align}\label{gt}
			{g_{\bf{s}}}\left( {{{\bf{\Theta }}^\alpha }} \right) =\sum\limits_{p = 1}^{P({\bf{s}})}  \sum\limits_{q = 1}^{{L_p}({\bf{s}})}  &\alpha _{{\bf{s}},{l_{p,q}}}^{}\delta \left( {\theta  - {\theta _{{\bf{s}},{l_{p,q}}}}} \right)\delta \left( {\beta  - {\beta _{{\bf{s}},{l_{p,q}}}}} \right)\notag\\
			&\delta \left( {\tau  - {\tau _{{\bf{s}},{l_{p,q}}}}} \right),
		\end{align}
		the CTF expression in (\ref{CTF1}) can be reformulated as
		\begin{align}\label{CTF2}
			\!\!\!	{h_{a,b}}\left( {n{T_{{\rm{sym}}}},k\Delta f} \right) = \int {\int {\int {{g_{\bf{s}}}\left( {{{\bf{\Theta }}^\alpha }} \right){e^{\bar \jmath 2\pi \nu n{T_{\rm{sym}}}}}} } } {e^{-\bar \jmath 2\pi k\Delta f\tau }}\notag\\
			\!\!\!\!\!\!\!\!\!\!\!\!\!\!\!\!\!\!\!\!\!\!\!\!\!\!\!{e^{ - \bar \jmath 2\pi \left( {{f_{\rm{c}}} + k\Delta f} \right)\frac{{(a-1){d_{{\rm{BS}}}}\cos (\theta ) + (b-1){d_{{\rm{UT}}}}\cos (\beta )}}{c}}}d\theta d\beta d\tau,
		\end{align}
		where the Doppler shift is $\nu = {v_{{\rm{move}}}} \cos({\alpha_{{\rm{move}}}} - \beta)/{\lambda_{\rm c}}$.

			The function ${g_{{{\bf{s}}}}}\left( {{\bf{\Theta }}^\alpha } \right) $ characterizes the channel behavior in the AoD-AoA-delay (AAD) domain at a given UT location ${\bf{s}}$. The CTF in (\ref{CTF1}) is fully determined by ${g_{{{\bf{s}}}}}\left( {{\bf{\Theta }}^\alpha } \right) $ that is solely governed by the propagation environment between  ${{{\bf{s}}_{{\rm{BS}}}}}$ and ${{{\bf{s}}}}$. To capture the statistical properties of the channel at each location, we define the function ${S_{{{\bf{s}}}}}\left( {{\bf{\Theta }}^\alpha } \right) = {\mathbb E}\left\{ {{{\left| {{g_{{{\bf{s}}}}}\left( {{\bf{\Theta }}^\alpha } \right)} \right|}^2}} \right\}$ for all ${\bf{s}} \in {\cal S}$, which describes the power distribution in the AAD domain for the channel between ${{\bf{s}}_{\rm BS}}$ and ${\bf{s}}$. In the next subsection, we investigate the relationship between $S_{\mathbf{s}}(\boldsymbol{\Theta}^\alpha)$ and the statistical parameters of each cluster, based on which we introduce the concept of the PCF. 

			\subsection{Physical Channel Fingerprints}
			
			The function ${S_{{{\bf{s}}}}}({\bf{\Theta }}^\alpha)$ statistically characterizes the propagation environment between the BS and ${\bf{s}}$. Its unified multi-domain representation of channel statistics holds the potential to be transformed into various forms of sCSI for different applications. This naturally motivates the construction of a location-indexed charting, where each entry corresponds to ${S_{{\bf{s}}}}({\bf{\Theta}}^\alpha)$ that serves as the PCF for ${\bf{s}}$.

			However, as a multi-dimensional function, directly storing ${S_{{{\bf{s}}}}}\left( {{\bf{\Theta }}^\alpha } \right)$ for different ${\bf{s}} \in {\cal S}$ to form a channel charting with PCFs is impractical for applications. Instead, we can fully exploit the statistical characteristics of the intra-cluster physical parameters to reformulate ${S_{{{\bf{s}}}}}\left( {{\bf{\Theta }}^\alpha } \right)$. By adopting suitable functional representations, the delay PDF and angle PDFs of MPCs within the same cluster can be effectively modeled \cite{9786750, 9318511,10225614, huang2020geometry, 10288151}. As a result, ${S_{\bf{s}}}\left( {{{\bf{\Theta }}^\alpha }} \right)$ can be approximated as
			\begin{equation}\label{St}
				{S_{\bf{s}}}\left( {{{\bf{\Theta }}^\alpha }} \right) \approx \sum\limits_{p = 1}^{P({\bf{s}})} { {{\gamma _{{\bf{s}},p}}f_{{\bf{s}},p}^{\rm{\upbeta }}\left( \beta  \right)f_{{\bf{s}},p}^{\rm{\uptheta }}\left( \theta  \right)f_{{\bf{s}},p}^{\rm{\uptau }}\left( \tau  \right)}},
			\end{equation}
			where ${\gamma _{{\bf{s}},p}} = {\mathbb{E}}\left\{ {\sum\nolimits_{q = 1}^{{L_p}({\bf{s}})} {|\alpha _{{\bf{s}},{l_{p,q}}}^{}{|^2}} } \right\}$ represents the average power of the $p$-th cluster, ${f_{{\bf{s}},p}^{\rm{\upbeta }}\left( \beta  \right)}$, ${f_{{\bf{s}},p}^{\rm{\uptheta }}\left( \theta  \right)}$ and ${f_{{\bf{s}},p}^{\rm{\uptau }}\left( \tau  \right)}$ represent the PDFs of the AoD, AoA and delay of MPCs in the $p$-th cluster, respectively. All the PDFs involved in (\ref{St}) can typically be approximated using a unimodal symmetric function, such as the PDF of Gaussian distribution \cite{9318511,9786750}, Laplacian distribution \cite{huang2020geometry} or von Mise distribution \cite{10288151}.

			By employing the expression in (\ref{St}) to characterize the AAD domain power distribution, the channel charting with PCFs can be efficiently compressed by storing only the key parameters of the underlying PDFs, rather than the full representation of ${S_ {{{\bf{s}}}}}\left( {{\bf{\Theta }}^\alpha } \right)$. We can adopt the PDFs of Gaussian distributions to fit ${f_{{\bf{s}},p}^{\rm{\upbeta }}\left( \beta  \right)}$, ${f_{{\bf{s}},p}^{\rm{\uptheta }}\left( \theta  \right)}$ and ${f_{{\bf{s}},p}^{\rm{\uptau }}\left( \tau  \right)}$ for different clusters. For instance, the AoD PDF, ${f_{{\bf{s}},p}^{\rm{\upbeta }}\left( \beta  \right)}$, can be approximated as 
			\begin{equation}\label{Gaussian}
				f_{{\bf{s}},p}^{\rm{\upbeta }}\left( \beta  \right) = \frac{1}{{\sqrt {2\pi } \sigma _{{\bf{s}},p}^{\rm{\upbeta }}}}\exp \left( { - {{(\beta  - {{\bar \beta }_{{\bf{s}},p}})}^2}/2{{(\sigma _{{\bf{s}},p}^{\rm\upbeta} )}^2}} \right),
			\end{equation}
			where ${{{\bar \beta }_{{\bf{s}},p}}}$ and ${\sigma _{{\bf{s}},p}^{\rm{\upbeta }}}$ are the mean AoD and the standard derivation of the Gaussian distribution, respectively. Similarly, ${f_{{\bf{s}},p}^{\rm{\uptheta }}\left( \theta  \right)}$ and ${f_{{\bf{s}},p}^{\rm{\uptau }}\left( \tau  \right)}$ can be characterized by the parameters $\{ {{\bar \theta }_{{\bf{s}},p}},\sigma _{{\bf{s}},p}^\uptheta \} $ and $\{ {{\bar \tau }_{{\bf{s}},p}},\sigma _{{\bf{s}},p}^\uptau \} $, respectively. It is worth noting that, due to the inherent non-negativity of delay, the delay PDF is more accurately modeled as a truncated Gaussian distribution over the interval $[0, +\infty)$ \cite{huang2020geometry}. 

			Given the parameters $\{ {{\bar \beta }_{{\bf{s}},p}},\sigma _{{\bf{s}},p}^\upbeta \} $, $\{ {{\bar \theta }_{{\bf{s}},p}},\sigma _{{\bf{s}},p}^\uptheta \} $ and $\{ {{\bar \tau }_{{\bf{s}},p}},\sigma _{{\bf{s}},p}^\uptau \} $, and ${{\gamma _{{\bf{s}},p}}}$ for $p = 1,2, \cdots ,P({{\bf{s}}})$, the AAD domain power distribution corresponding to the location  ${\bf{s}}$  can be reconstructed. In this case, we define the vector
			\begin{align}\label{SCF}
				{\bf{m}}_{{{\bf{s}},p}} = [{{\bar \beta }_{{\bf{s}},p}},\sigma _{{\bf{s}},p}^\upbeta, {{\bar \theta }_{{\bf{s}},p}},\sigma _{{\bf{s}},p}^\uptheta, {{\bar \tau }_{{\bf{s}},p}},\sigma _{{\bf{s}},p}^\uptau,{{\gamma _{{\bf{s}},p}}}{]^{\rm{T}}},
			\end{align} 
			and the set ${{\cal M}_{{{\bf{s}}}}} = \{ {\bf{m}}_{{{\bf{s}},p}}\} _{p = 1}^{P({{\bf{s}}})}$. Then ${{\cal M}_{{{\bf{s}}}}}$ can fully characterize the physical property of the channel between ${{{\bf{s}}_{{\rm{BS}}}}}$ and ${{{\bf{s}}}}$. Furthermore, given the UT mobility and location, ${{\cal M}_{{{\bf{s}}}}}$ can also be used to analyze the channel behavior in the Doppler domain. By using ${{\cal M}_{{{\bf{s}}}}}$ in place of ${S_{{{\bf{s}}}}}\left( {{\bf{\Theta }}^\alpha } \right)$ as the PCF at location ${\bf{s}}$, the storage burden of the channel charting with PCFs is significantly alleviated. As such, ${{\cal M}_{{{\bf{s}}}}}$ satisfies both the physics-driven property and storage efficiency outlined in Section~\ref{section1}, and is adopted as the PCF throughout this work. Its practical benefits and applications will be discussed in subsequent sections.

			The aforementioned PCF is specifically defined for the scenario illustrated in Fig. \ref{txtorx1}. Nevertheless, its physics-driven property provides a natural foundation for extension to more realistic 6G propagation scenarios. Specifically, the proposed PCF is inherently compatible with different forms of channel non-stationarity, including spatial, frequency-dependent, and longer-term environmental variations (e.g., seasonal changes). The channel charting can be further extended from a static location-indexed representation to that incorporating the dynamic wireless environments. In all these scenarios, the core principle remains unchanged: the PCF compactly captures the essential physical characteristics of the propagation environment using a scalable and physics-interpretable parameter set.

			\section{PCF Acquisition in Massive MIMO-OFDM Systems}\label{section3}
			
			In this section, we propose a methodology for PCF acquisition in massive MIMO-OFDM systems. We begin by formulating the SFT domain channel model and establishing its relationship with the physical channel parameters. An uplink received signal model is then constructed to enable SFT domain channel estimation. From the estimated channels, we apply a SAGE-based algorithm to extract PCFs at different UT locations. By aggregating these location-specific PCFs, a structured channel charting with PCFs can be constructed, providing a location-indexed representation of the physical channel environment.

			


			\subsection{SFT Domain Channel Model}\label{section3a}
			
			\begin{figure}[t!]
				\centering
				\includegraphics[width =180pt]{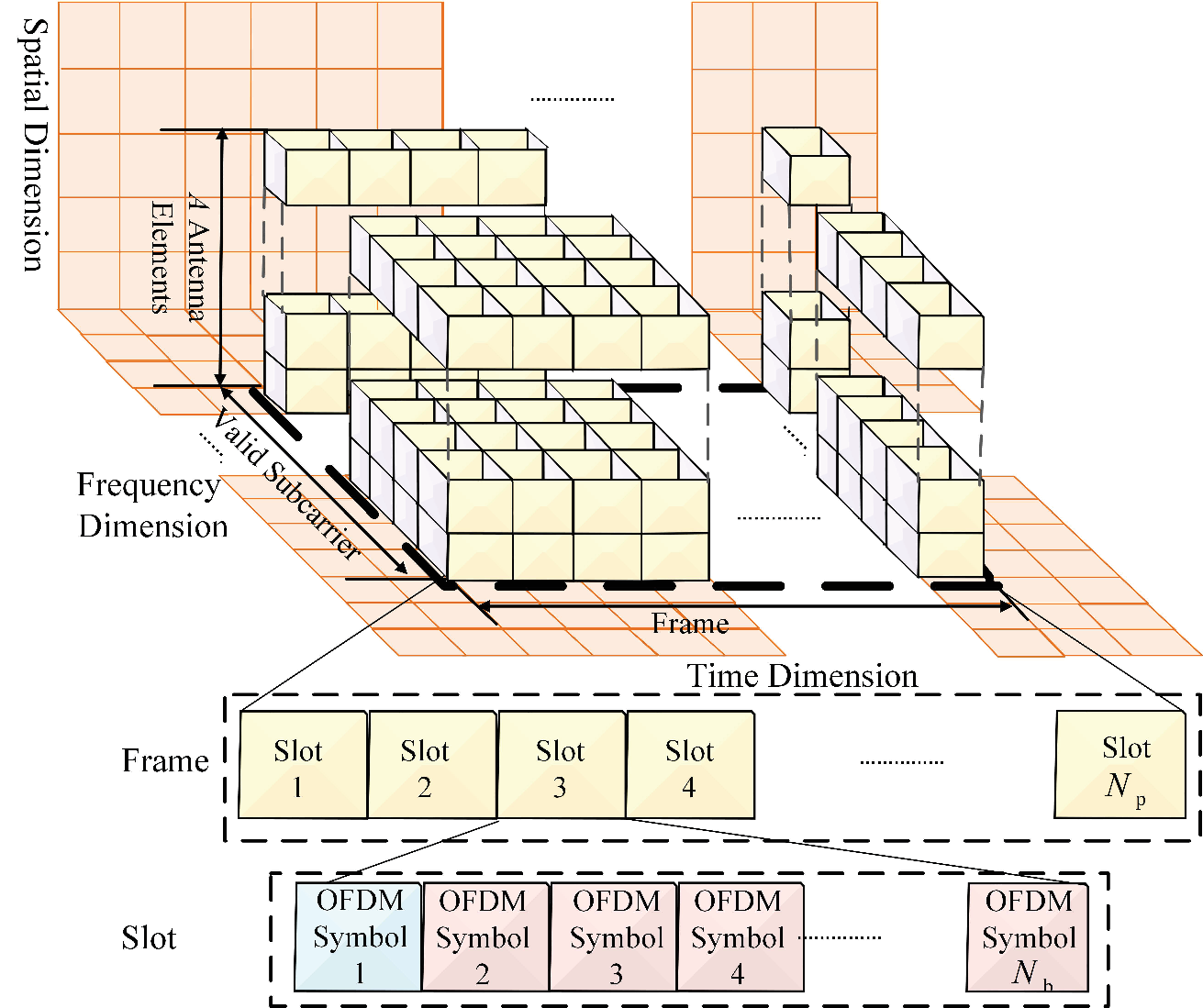}
				\caption{ The structure of the SFT domain channel tensor for the special case of $B = 1$.}
				\label{frame11}
			\end{figure}
			We consider the scenario depicted in Fig.\ref{txtorx1}, where ULAs with dimensions of $A \times 1$ and $B \times 1$ are configured at the BS and UT sides, respectively. Along the subcarrier dimension of the OFDM system, we select $K$ valid subcarriers for transmitting pilot signals, and denote the set of indices of these valid subcarriers as ${\cal K} = \{ {k_0},{k_0} + 1, \cdots ,{k_0} + K - 1\}$. 
			

			Along the time dimension, we assume that the UT moves within a single grid $\mathbf{s}$ during the analysis period. This period is divided into multiple frames, each comprising $N_{\mathrm{s}}$ OFDM symbols. As an illustrative example, Fig.~\ref{frame11} depicts the structure of the SFT domain channel within a single frame for the special case of $B = 1$. As depicted, each frame is further partitioned into $N_{\mathrm{p}}$ slots, with each slot comprising $N_{\mathrm{b}}$ OFDM symbols, such that ${N_{\mathrm{s}}} = {N_{\mathrm{p}}}{N_{\mathrm{b}}}$. Within each slot, the first OFDM symbol (highlighted in light blue) is allocated for uplink pilot transmission to facilitate channel estimation, while the remaining symbols (highlighted in light red) are used for uplink and downlink data transmission. Under this structure, there are $N_{\mathrm{p}}$ pilot segments per frame, and the interval between two pilot transmissions equals the duration of one slot.
			

			To generate the PCF for a given location $\mathbf{s}$ under a MIMO-OFDM configuration, we first establish the SFT domain channel model between the BS and $\mathbf{s}$, and clarify the relationship between the SFT domain channel and the PCF. Our analysis is conducted on a frame-by-frame basis. Let $\tilde{n}$ denote the index of the current slot. The OFDM symbol index is counted from 0 along the time dimension. From the frame structure in Fig.~\ref{frame11}, the $\tilde{n}$-th slot, along with its preceding $({N_{\mathrm{p}}} - 1)$ slots, forms the current frame. With ${t_{\tilde n,n}} \buildrel \Delta \over = \tilde n - {N_{\rm{p}}} + n$, the set of indices for all OFDM symbols in the pilot segments of the current frame is denoted as ${{\cal T}_{\tilde n}} = \{ {t_{\tilde n,1}}{N_{\rm{b}}},{t_{\tilde n,2}}{N_{\rm{b}}}, \cdots ,{t_{\tilde n,{N_{\rm{p}}}}}{N_{\rm{b}}}\} $. We then define an ${A \times K \times {N_{\mathrm{p}}} \times B}$ tensor, denoted as ${{\mathbfcal H}_{{\bf{s}},\tilde{n}}^{{\rm{SFT}}}}$, where
			\begin{equation}\label{htensor}
				{[{\mathbfcal H}_{{\bf{s}},\tilde n}^{{\rm{SFT}}}]_{a,k,n,b}} = {h_{a,b}}\left( {{t_{\tilde n,n}}{N_{\rm{b}}}{T_{{\rm{sym}}}},k\Delta f} \right),
			\end{equation}
			to represent the SFT domain channel between the BS and UT across all valid subcarriers in the pilot segments of the current frame. The definition of ${h_{a,b}}\left( \cdot \right)$ is provided in (\ref{CTF1}). Let ${\bf{h}}_{{\bf{s}},\tilde{n}}^{{\rm{SFT}}} \buildrel \Delta \over = {\rm{vec\{ }}{\mathbfcal H}_{{\bf{s}},\tilde{n}}^{{\rm{SFT}}}{\rm{\} }}$ denote the vectorization of ${\mathbfcal H}_{{\bf{s}},\tilde{n}}^{{\rm{SFT}}}$. With (\ref{CTF1}), ${\bf{h}}_{{\bf{s}},\tilde{n}}^{{\rm{SFT}}}$ can also be calculated by
			\begin{equation}\label{hsft}
				{\bf{h}}_{{\bf{s}},\tilde n}^{{\rm{SFT}}}{\rm{ = }}\sum\limits_{p = 1}^{P({\bf{s}})} {\sum\limits_{q = 1}^{{L_p}({\bf{s}})} {\alpha _{{\bf{s}},{l_{p,q}}}^{}{\bf{v}}_{\tilde n}^{{\rm{sft}}}(\beta _{{\bf{s}},{l_{p,q}}}^{},\theta _{{\bf{s}},{l_{p,q}}}^{},{\tau _{{\bf{s}},{l_{p,q}}}})} },
			\end{equation}
			where the vector ${{\bf{v}}_{\tilde n}^{{\rm{sft}}}(\beta _{{\bf{s}},{l_{p,q}}}^{},\theta _{{\bf{s}},{l_{p,q}}}^{},{\tau _{{\bf{s}},{l_{p,q}}}})}$ satisfies
			\begin{align}\label{vsft}
				\!\!\!\!\!\!\!&{[{\bf{v}}_{\tilde n}^{{\rm{sft}}}(\beta ,\theta ,\tau )]_{bAK{N_{\rm{p}}}{\rm{ + }}nAK + kA + a}} = \scalebox{0.8}{$\exp \left\{ {\bar \jmath 2\pi \left[ -{k\Delta f\tau } \right.} \right.$}\notag\\
				&\scalebox{0.78}{$\left. { + {\frac{{{v_{{\rm{move}}}}\cos \left( {{\alpha _{{\rm{move}}}} - \beta } \right){t_{\tilde n,n}}{N_{\rm{b}}}{T_{{\rm{sym}}}}}}{{\lambda}_c }} - \left( {{f_{\rm{c}}} + k\Delta f} \right)\left. {\frac{{(a-1){d_{{\rm{BS}}}}\cos (\theta ) + (b-1){d_{{\rm{UT}}}}\cos (\beta )}}{c}} \right]} \right\}.$}
			\end{align}
			It should be noted that in (\ref{vsft}), the Doppler shift value is replaced by $\nu  = {v_{{\rm{move}}}}\cos \left( {{\alpha _{{\rm{move}}}} - \beta } \right)/{\lambda_c}$.
			
			Define the following vectors: ${\bf{v}}_{\rm{\uptheta }}^{\rm{s}}\left( \theta  \right) \in {{\mathbb{C}}^{A \times 1}}$, where ${\left[ {{\bf{v}}_{\rm{\uptheta }}^{\rm{s}}\left( \theta  \right)} \right]_a} = {e^{ - \bar \jmath 2\pi (a - 1){f_{\rm{c}}}\frac{{{d_{{\rm{BS}}}}}}{c}\cos \theta }}$; ${\bf{v}}_{}^{\rm{f}}\left( \tau  \right) \in {{\mathbb{C}}^{K \times 1}}$, where ${\left[ {{\bf{v}}_{}^{\rm{f}}\left( \tau  \right)} \right]_k} = {e^{ - \bar \jmath 2\pi ({k_0} + k - 1)\Delta f\tau }}$; ${\bf{v}}_{\rm{\upbeta }}^{\rm{s}}\left( \beta  \right) \in {{\mathbb{C}}^{B \times 1}}$ and ${\left[ {{\bf{v}}_{\rm{\upbeta }}^{\rm{s}}\left( \beta  \right)} \right]_b} = {e^{ - \bar \jmath 2\pi (b - 1){f_{\rm{c}}}\frac{{{d_{{\rm{UT}}}}}}{c}\cos \beta }}$; ${\bf{v}}_{\upbeta ,\tilde n}^{\rm{t}}\left( \theta  \right) \in {{\mathbb{C}}^{{N_{\rm{p}}} \times 1}}$, where ${\left[ {{\bf{v}}_{\upbeta ,\tilde n}^{\rm{t}}\left( \beta  \right)} \right]_n} = {e^{\bar \jmath 2\pi {t_{\tilde n,n}}{N_{\rm{b}}}{T_{{\rm{sym}}}}\frac{{{v_{{\rm{move}}}}}}{\lambda _c }\cos \left( {\beta  - {\alpha _{{\rm{move}}}}} \right)}}$. Specifically, when the
			spatial wideband effect of the channel can be ignored, ${\bf{v}}_{\tilde n}^{{\rm{sft}}}(\beta ,\theta ,\tau )$ can be expressed as \cite{10179246}
			\begin{equation}
				{\bf{v}}_{\tilde n}^{{\rm{sft}}}(\beta ,\theta ,\tau ) = {\bf{v}}_{\rm{\uptheta }}^{\rm{s}}\left( \theta  \right) \otimes {\bf{v}}_{}^{\rm{f}}\left( \tau  \right) \otimes {\bf{v}}_{\upbeta ,\tilde n}^{{\rm{st}}}\left( \beta  \right),
			\end{equation}
			and the SFT domain channel model in (\ref{hsft}) simplifies to 
			\begin{align}\label{hsft1}
				{\bf{h}}_{{\bf{s}},\tilde n}^{{\rm{SFT}}}\!\! = \!\!\scalebox{0.95}{$\sum\limits_{p = 1}^{P({\bf{s}})} {\sum\limits_{q = 1}^{{L_p}({\bf{s}})}  \alpha _{{\bf{s}},{l_{p,q}}}^{}{\bf{v}}_{\rm{\uptheta }}^{\rm{s}}\left( {{\theta _{{\bf{s}},{l_{p,q}}}}} \right) \otimes {\bf{v}}_{}^{\rm{f}}\left( {{\tau _{{\bf{s}},{l_{p,q}}}}} \right) \otimes {\bf{v}}_{\upbeta ,\tilde n}^{{\rm{st}}}\left( {{\beta _{{\bf{s}},{l_{p,q}}}}} \right)},$}
			\end{align} 
			where  ${\bf{v}}_{\upbeta ,\tilde n}^{{\rm{st}}}\left( \beta  \right) = {\bf{v}}_{\upbeta ,\tilde n}^{\rm{t}}\left( \beta  \right) \otimes {\bf{v}}_{\rm{\upbeta }}^{\rm{s}}\left( \beta  \right)$.

			Equations (\ref{hsft}) and (\ref{hsft1}) establish the relationship between the SFT domain channel and the physical parameters $\{ \alpha _{{\bf{s}},{l_{p,q}}}^{},{\beta _{{\bf{s}},{l_{p,q}}}},{\theta _{{\bf{s}},{l_{p,q}}}},{\tau _{{\bf{s}},{l_{p,q}}}}\}$ of different MPCs. Since PCFs characterize the statistical properties of these parameters, estimating the physical parameters of individual MPCs is a prerequisite for PCF acquisition.

			Given that the clustering of the channel is unknown prior to estimating the MPC parameters, we rewrite (\ref{hsft}) and express ${\bf{h}}_{{\bf{s}},\tilde{n}}^{{\rm{SFT}}}$ as 
			\begin{equation}
				{\bf{h}}_{{\bf{s}},\tilde n}^{{\rm{SFT}}}{\rm{ = }}\sum\limits_{l = 1}^{L({\bf{s}})} {\alpha _{{\bf{s}},l}^{}{\bf{v}}_{\tilde n}^{{\rm{sft}}}({\beta _{{\bf{s}},l}},{\theta _{{\bf{s}},l}},{\tau _{{\bf{s}},l}})},
			\end{equation}
			where $\{ \alpha _{{\bf{s}},l}^{},{\beta _{{\bf{s}},l}},{\theta _{{\bf{s}},l}},{\tau _{{\bf{s}},l}}\}$ are the parameters to be estimated of the $l$-th MPC.

			In practical operations, the tensor ${\mathbfcal H}_{{\bf{s}},\tilde{n}}^{{\rm{SFT}}}$ (or its vectorized form ${\bf{h}}_{{\bf{s}},\tilde{n}}^{{\rm{SFT}}}$) is obtained through uplink measurements \cite{10075521}. Specifically, the UT moves in the vicinity of $\mathbf{s}$ at a constant speed $v_{\mathrm{move}}$ along a fixed direction $\alpha_{\mathrm{move}}$, while all its antennas simultaneously transmit uplink pilot signals to the BS. The BS then estimates ${\mathbfcal H}_{{\bf{s}},\tilde{n}}^{{\rm{SFT}}}$ from the received signals. By performing joint multi-symbol channel estimation on a frame-by-frame basis, the received signal at the BS for the current frame can be expressed as
			\begin{equation}
				{\mathbfcal Y}{\rm{ = }}\sum\limits_{b = 1}^B {{\bf{X}}_{b,\tilde n}^{\rm{t}}{ \times _3}\left( {{\bf{X}}_b^{\rm{f}}{ \times _2}{{[{\mathbfcal H}_{{\bf{s}},\tilde n}^{{\rm{SFT}}}]}_{:,:,:,b}}} \right)}   + {\mathbfcal Z
				},
			\end{equation}
			where ${\mathbfcal Y} \in {{\mathbb{C}}^{A \times K \times {N_{\rm{p}}}}}$ and ${[{\mathbfcal Y}]_{a,k,n}}$ is the received signal of the $a$-th antenna element at the BS on the $k$-th OFDM signal of the $n$-th pilot segment. ${\bf{X}}_{b,\tilde n}^{\rm{t}} \in {{\mathbb{C}}^{{N_{\rm{p}}} \times {N_{\rm{p}}}}}$ and ${\bf{X}}_b^{\rm{f}} \in {{\mathbb{C}}^{K \times K}}$ are all diagonal matrices as ${\bf{X}}_b^{\rm{f}} = {\rm{Diag\{ }}x_{b,{k_0}}^{\rm{f}},x_{b,{k_0} + 1}^{\rm{f}}, \cdots ,x_{b,{k_0} + K - 1}^{\rm{f}}{\rm{\} }}$ and ${\bf{X}}_{b,\tilde n}^{\rm{t}} = {\rm{Diag}}\{ x_{b,{t_{\tilde n,1}}}^{\rm{t}},x_{b,{t_{\tilde n,2}}}^{\rm{t}}, \cdots ,x_{b,{t_{\tilde n,{N_{\rm{p}}}}}}^{\rm{t}}\}$. And $x_{b,n}^{\rm{t}}x_{b,{k}}^{\rm{f}}$ is the pilot signal transmitted from the $b$-th antenna elements at the UT on the $k$-th subcarrier of the $n$-th OFDM symbol. 
			
			Given its superior properties for SFT domain channel estimation, we employ the two-dimensional time-frequency phase-shifted pilots (TFPSPs) as the pilot signals \cite{mine,10693363,10146318}. TFPSPs are typically constructed by modulating Zadoff-Chu (ZC) sequences with phase-shifted sequences, following the general format:
			\begin{subequations}\label{tfpsp}
				\begin{align}
					&x_{b,k}^{\rm{f}} = {e^{ - \bar \jmath 2\pi k{\phi _b}/K}} \cdot {e^{2\pi {\gamma _{\rm{f}}}\frac{{k\left( {k + {{\left\langle K \right\rangle }_2}} \right)}}{K}}},k \in {\cal K},\\
					&x_{b,n}^{\rm{t}} = {e^{ - \bar \jmath 2\pi n{\varphi _b}/{N_{\rm{p}}}}} \cdot {e^{2\pi {\gamma _{\rm{t}}}\frac{{n\left( {n + {{\left\langle {{N_{\rm{p}}}} \right\rangle }_2}} \right)}}{{{N_{\rm{p}}}}}}},n \in {{\mathbb{N}}^ + },
				\end{align}
			\end{subequations}
			where ${{\phi _b}}$ and ${{\varphi _b}}$ represent the phase-shifted factors in the frequency domain and time domain, respectively. ${{\gamma _{\rm{f}}}}$ and ${{\gamma _{\rm{t}}}}$ denote the roots of the frequency-domain and time-domain ZC sequences, respectively, which are chosen as positive integers coprime with $K$ and ${{N_{\rm{p}}}}$, respectively. Let ${\nu _{\max }}{\rm{ = 2}}{v_{{\rm{move}}}}/{\lambda _{\rm{c}}}$, ${N_{\rm{f}}} = \left\lceil {K{N_{\rm{g}}}/{N_{\rm{c}}}} \right\rceil$, and ${N_{\rm{d}}} = \left\lceil {{\nu _{\max }}{T_{{\rm{sym}}}}{N_{\rm{s}}}} \right\rceil$. According to the property of TFPSPs in \cite{mine}, the pilot sequences associated with different transmit antennas become mutually non-interfering when their phase-shifted factors are sufficiently separated. Specifically, for any two distinct antenna elements $b$ and $b'$, if the phase-shifted factors satisfy
			\begin{equation}\label{tfpsp1}
				\left| {{\phi _b} - {\phi _{b'}}} \right| \ge {N_{\rm{f}}},
			\end{equation}
			or
			\begin{equation}\label{tfpsp2}
				\left| {{\varphi _b} - {\varphi _{b'}}} \right| \ge {N_{\rm{d}}},
			\end{equation}
			the pilot interference among antenna elements can be eliminated.
			
			Under the conditions specified in (\ref{tfpsp1}) and (\ref{tfpsp2}), let ${K_{\rm{f}}} = \left\lfloor {K{\rm{/}}{N_{\rm{f}}}} \right\rfloor$ and ${K_{\rm{t}}} = \left\lfloor {{N_{\rm{p}}}/{N_{\rm{d}}}} \right\rfloor$. Then, when $B \le {K_{\rm{f}}}{K_{\rm{t}}}$, for the $b$-th antenna at the UT, the phase-shifted factors can be set as
			\begin{equation}
				{\phi _b} = {\left\langle {b - 1} \right\rangle _{{K_{\rm{f}}}}}{N_{\rm{f}}},\quad {\varphi _b} = \left\lfloor {(b - 1)/{K_{\rm{f}}}} \right\rfloor {N_{\rm{d}}},
			\end{equation}
			to satisfy either (\ref{tfpsp1}) or (\ref{tfpsp2}). When $B > {K_{\rm{f}}}{K_{\rm{t}}}$, the pilot capacity can be expanded by generating distinct values of ${{\gamma _{\rm{f}}}}$ and ${{\gamma _{\rm{t}}}}$ \cite{10146318}. With the above pilot settings, we can transmit TFPSPs and ultilize the method in \cite{10839435} to estimate ${{\mathbfcal H}_{{\bf{s}},\tilde n}^{{\rm{SFT}}}}$ from the received ${\mathbfcal Y}$. 
			
			\subsection{Generating PCF from the Estimated SFT Domain Channel}
			
			
			By applying the pilot design and channel estimation procedure outlined in the previous subsection, we obtain an estimate of ${{\mathbfcal H}_{{\bf{s}},\tilde n}^{{\rm{SFT}}}}$, denoted as $\hat {\mathbfcal H}_{{\bf{s}},\tilde n}^{{\rm{SFT}}}$. Subsequently, a SAGE-based method can be employed to estimate the MPC parameters with improved accuracy.

			Let ${\bf{\hat h}}_{{\bf{s}},\tilde n}^{{\rm{SFT}}} = {\rm{vec}}\left\{ {\hat {\mathbfcal H}_{{\bf{s}},\tilde n}^{{\rm{SFT}}}} \right\}$. Before extracting MPC parameters from
			${\bf{\hat h}}_{{\bf{s}},\tilde n}^{{\rm{SFT}}}$, the number of MPCs is typically unknown. Therefore, we predefine a value $\hat L({{\bf{s}}})$ as the expected number of MPCs in the scenario. For the MPC parameter  ${\bf{\Theta }} = [\alpha ,\beta ,\theta ,\tau]^{\rm T}$, we define the following function:
			\begin{equation}
				{{\bf{h}}_{\tilde n}}({\bf{\Theta }}) = \alpha {\bf{v}}_{\tilde n}^{{\rm{sft}}}(\beta ,\theta ,\tau ).
			\end{equation} 
			Then under the maximum likelihood (ML) criterion \cite{10494205}, the MPC parameter estimation problem can be interpreted as finding the MPC parameters ${{{\bf{\hat \Phi }}}_{\bf{s}}} = [{{{\bf{\hat \Theta }}}_{{\bf{s}},1}},{{{\bf{\hat \Theta }}}_{{\bf{s}},2}}, \cdots ,{{{\bf{\hat \Theta }}}_{{\bf{s}},\hat L({\bf{s}})}}]$, such that
			\begin{equation}\label{pro1}
				{{{\bf{\hat \Phi }}}_{\bf{s}}} = \scalebox{0.9}{$\mathop {\arg \min }\limits_{[{{\bf{\Theta }}_{{\bf{s}},1}},{{\bf{\Theta }}_{{\bf{s}},2}}, \cdots ,{{\bf{\Theta }}_{{\bf{s}},\hat L({\bf{s}})}}]} \left\| {{\bf{\hat h}}_{{\bf{s}},\tilde n}^{{\rm{SFT}}} - \sum\limits_{l = 1}^{\hat L({\bf{s}})} {{{\bf{h}}_{\tilde n}}({{\bf{\Theta }}_{{\bf{s}},l}})} } \right\|_{\rm{F}}^2.$}
			\end{equation}


			In massive MIMO-OFDM systems, the assumption of uncorrelated scattering is approximately valid \cite{you2015pilot}, and ${{\bf{h}}_{\tilde n}}({{\bf{\Theta }}_{{\bf{s}},l}})$ and ${{\bf{h}}_{\tilde n}}({{\bf{\Theta }}_{{\bf{s}},l'}})$ can be approximately considered orthogonal for $l \ne l'$. This enables the direct application of the SAGE algorithm, which is a variant of the expectation-maximization (EM) algorithm that employs successive interference cancellation (SIC) \cite{753729,sage11,10179246}. The SAGE algorithm is adopted for its distinct advantages: it performs full MPC parameter estimation, is not constrained by array geometry or snapshot requirements, and provides a computationally efficient and scalable solution suitable for large-dimensional systems. Accordingly, the SAGE algorithm transforms the computationally intensive global search for ${{{\bf{\hat \Phi }}}_{\bf{s}}}$ in (\ref{pro1}) into a sequential estimation of ${{{\bf{\hat \Theta }}}_{{\bf{s}},1}},{{{\bf{\hat \Theta }}}_{{\bf{s}},2}}, \cdots ,{{{\bf{\hat \Theta }}}_{{\bf{s}},\hat L({\bf{s}})}}$. In the parameter estimation for the $l$-th MPC, we define
			\begin{equation}\label{ykl}
				\scalebox{0.9}{${\bf{\hat h}}_{{\bf{s}},\tilde n,l}^{{\rm{SFT}}} = \left\{ {\begin{array}{*{20}{c}}
							{\!\!\quad{\bf{\hat h}}_{{\bf{s}},\tilde n}^{{\rm{SFT}}},\quad l = 1}\\
							{{\bf{\hat h}}_{{\bf{s}},\tilde n}^{{\rm{SFT}}} - \sum\limits_{l' = 1}^{l - 1} {{{\bf{h}}_{\tilde n}}({{{\bf{\hat \Theta }}}_{{\bf s},l}})}  ,l = 2,3, \cdots ,\hat L({\bf{s}}),}
					\end{array}} \right.$}
			\end{equation}
			and the aim of the proposed SAGE-based method is to find a parameter set ${{{\bf{\hat \Theta }}}_{{\bf{s}},l}}$ satisfying 
			\begin{equation}\label{pro20}
				\scalebox{0.94}{${{{\bf{\hat \Theta }}}_{{\bf{s}},l}} = \mathop {\arg \min }\limits_{{{\bf{\Theta }}_{{\bf{s}},l}}} \left\| {{\bf{\hat h}}_{{\bf{s}},\tilde n,l}^{{\rm{SFT}}} - {{\bf{h}}_{\tilde n}}({{\bf{\Theta }}_{{\bf{s}},l}})} \right\|_{\rm{F}}^2$}
			\end{equation}
			until ${\hat L({{\bf{s}}})}$ MPC parameter vectors are estimated.

			
			Under the uncorrelated scattering assumption, by reformulating the original problem in (\ref{pro1}) into the equivalent problem (\ref{pro20}), the latter can be efficiently solved using the wideband channel parameter estimation algorithm proposed in \cite{10179246}. Specifically, when the SFT channel is expressed in the form of (\ref{hsft1}), it can be demonstrated that the parameters $\alpha _{{\bf{s}},l}^{},{\beta _{{\bf{s}},l}},{\theta _{{\bf{s}},l}},{\tau _{{\bf{s}},l}}$ constitute mutually admissible data \cite{753729,sage11}. This property enables the utilization of the estimation framework in \cite{sage11} to obtain an efficient solution to (\ref{pro20}).

			Using the estimated MPC parameters, the ${\hat L({{\bf{s}}})}$ sets of MPC parameters are partitioned into $P$ groups, with each group corresponding to a distinct cluster.  With ${\gamma _{{\bf{s}},l}}{\rm{ = }}{\left| {{\alpha _{{\bf{s}},l}}} \right|^{\rm{2}}}$ as the power for the $l$-th MPC, the Kernel-Power-Density (KPD) based algorithm \cite{8013075, 8890833} can then be employed to cluster the MPCs according to their similarity. The value of $P$ can be determined by minimizing the Davies-Bouldin (DB) index \cite{8013075}. Specifically, using the MPC samples in each cluster, the functions defined in (\ref{St})-(\ref{Gaussian}) are applied to obtain a smoothed and physically meaningful PCF per cluster. This statistical averaging and model-based fitting process effectively filters out uncorrelated noise and ensures that the extracted PCFs capture the stable, large-scale scattering characteristics rather than instantaneous fluctuations. By moving a UT across different grids and collecting the corresponding PCFs, a channel charting with PCFs can be constructed to characterize the physical properties of the channels between the BS and the UTs within the whole cell coverage. When a new mobile UT enters the coverage area, the BS can retrieve the relevant PCF from the cell-specific channel charting based on the UT location and generate the required sCSI by incorporating the UT’s antenna configuration and mobility state. As the PCF-to-sCSI transformation operates on the smoothed PCFs, the resulting sCSI is inherently robust to measurement imperfections. This enables environment-aware communication without the need for online channel probing. The next section provides the detailed procedure for transforming a PCF into UT-specific sCSI. 

					\section{Statistical CSI Acquisition with PCFs}\label{section4}

					In this section, we propose a method for acquiring beam domain sCSI from PCFs. We first introduce a triple BBCM and establish its relationship with the GBSM. Then we derive the expression for beam-domain sCSI and develop a low-complexity algorithm to obtain it from the PCF.

					\begin{figure}[t!]
						\centering
						\includegraphics[width =180pt]{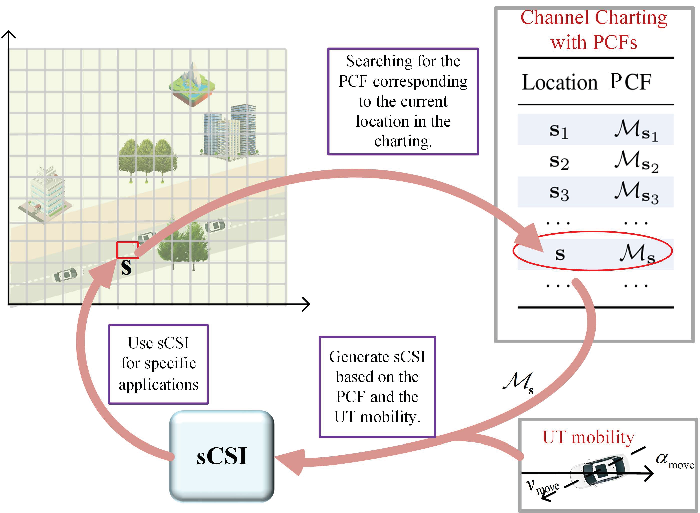}
						\caption{The procedure of sCSI acquisition with PCFs.}
						\label{frame1}
					\end{figure}

					\subsection{BBCM}

					We first introduce the concept of BBCM. For simplicity, we consider the case where the UT is equipped with a single antenna, i.e., $B=1$, and the BS is configured with an  $A \times 1$ ULA, with the element spacing of half a wavelength, i.e., ${d_{{\rm{BS}}}} = 0.5{\lambda}_{\rm c} $. However, if more complex antenna configuration is considered at both the BS and UT sides, the channel model can be analyzed in a similar method.

					Assume that the UT moves within the grid ${\bf s}$ over a period of time, with the speed ${v_{{\rm{move}}}}$ towards the direction ${{\alpha}_{{\rm{move}}}}$. According to (\ref{CTF1}),  the CTF between the UT and the $a$-th antenna element at the BS on the $k$-th subcarrier of the $n$-th symbol can be rewritten as   
					\begin{align}\label{CTF5}
						{h_a}\left( {n{T_{{\rm{sym}}}},k\Delta f} \right) &= \sum\limits_{p = 1}^{P({\bf{s}})} {\sum\limits_{q = 1}^{{L_p}({\bf{s}})} {\alpha _{{\bf{s}},{l_{p,q}}}^{}{e^{\bar \jmath 2\pi {\nu _{{\bf{s}},{l_{p,q}}}}n{T_{{\rm{sym}}}}}}} }  \notag\\
						&\!\!\!\!\!\!{{e^{ - \bar \jmath 2\pi \left( {{f_{\rm{c}}} + k\Delta f} \right)(a-1)\frac{{{d_{{\rm{BS}}}}}}{c}{\vartheta _{{\bf{s}},{l_{p,q}}}}}}}{e^{-\bar \jmath 2\pi k\Delta f{\tau _{{\bf{s}},{l_{p,q}}}}}},
					\end{align}
					where the UT antenna index $b$ is omitted due to its single-antenna setting, the parameter ${\vartheta _{{\bf{s}},{l_{p,q}}}} = \cos {\theta _{{\bf{s}},{l_{p,q}}}}$ represents the angular cosine, and the Doppler shift is $\smash{{\nu _{{\bf{s}},{l_{p,q}}}} = {v_{{\rm{move}}}}\cos \left( {{\alpha _{{\rm{move}}}} - {\beta _{{\bf{s}},{l_{p,q}}}}} \right)/{\lambda _c}}$. We continue to adopt the frame structure illustrated in Fig.~\ref{frame11}. Then  for the $\tilde n$-th slot, the SFT domain channel between the BS and the UT in the current frame can be expressed as an $A \times K \times {N_{\rm{p}}}$ dimensional tensor ${{\mathbfcal H}_{{\bf{s}},\tilde{n}}^{{\rm{SFT}}}}$, where ${[{\mathbfcal H}_{{\bf{s}},\tilde n}^{{\rm{SFT}}}]_{a,k,n}} = {h_a}\left( {{t_{\tilde n,n}}{N_{\rm{b}}}{T_{{\rm{sym}}}},k\Delta f} \right)$.

					From the definition of angular cosine, we have ${\vartheta _{{\bf{s}},{l_{p,q}}}} \in {\cal B}_\upvartheta ^{}$, where ${\cal B}_\upvartheta ^{} = [-1, 1)$. With ${{\cal B}_\uptau } = [0,1/\Delta f)$, the delay satisfies $\tau_{{\bf{s}},{l_{p,q}}} \le {N_{\rm{g}}}{T_{\rm{s}}} < 1/\Delta f$, which implies $\tau_{{\bf{s}},{l_{p,q}}} \in {{\cal B}_\uptau }$. The Doppler shift $\nu_{{\bf{s}},{l_{p,q}}}$ lies within the range $[-{\nu_{\max}}/2, {\nu_{\max}}/2]$, where ${\nu_{\max}} \buildrel \Delta \over = 2{v_{\rm{move}}}/{\lambda_{\rm{c}}}$. For the frame structure, when $N_{\rm{b}} < 1/(T_{\rm{sym}} \nu_{\max})$, it follows that $\nu_{\max} < 1/(N_{\rm{b}} T_{\rm{sym}})$. Consequently, the range of $\nu_{{\bf{s}},{l_{p,q}}}$ can be further constrained to ${\cal B}_\upnu ^{}$, where ${\cal B}_\upnu ^{}=[-1/(2 N_{\rm{b}} T_{\rm{sym}}), 1/(2 N_{\rm{b}} T_{\rm{sym}})]$.

					To address the challenges of handling continuous domains in practical channel modeling, we discretize ${\cal B}_\upvartheta$, ${\cal B}_\uptau$ and ${\cal B}_\upnu$ into finite grids. Let the number of quantization points in the angular cosine, delay, and Doppler domains be ${N_\upvartheta}$, ${N_\uptau}$, and ${N_\upnu}$, respectively. The quantization points are defined as: ${{\vartheta} _{{n}}}{\rm{ }} \buildrel \Delta \over = \frac{{{n} -1- {N_\upvartheta }/2}}{{0.5{N_\upvartheta }}}$, ${{\tau} _n} \buildrel \Delta \over = \frac{n-1}{{{N_\uptau }\Delta f}}$, ${{\nu} _n} \buildrel \Delta \over = \frac{{\left( {n -1- {N_\upnu }/2} \right)}}{{{N_\upnu }{N_{\rm b}}{T_{{\rm{sym}}}}}}$. By collecting these quantization points into distinct sets, i.e., ${{\cal S}_\upvartheta } $ $\buildrel \Delta \over = $ $\left\{ {{\vartheta _1},} \right.$${\vartheta _2},$$\cdots,$$\left. {{\vartheta _{{N_\upvartheta } }}} \right\}$, ${{\cal S}_\uptau } $ $\buildrel \Delta \over = $ $\left\{ {{\tau _1},} \right.$${\tau _2},$$\cdots,$$\left. {{\tau _{{N_\uptau }}}} \right\}$ and ${{\cal S}_\upnu } $ $\buildrel \Delta \over = $ $\left\{ {{\nu _1},} \right.$${\nu _2},$$\cdots,$$\left. {{\nu _{{N_\upnu }}}} \right\}$, when ${N_\upvartheta}$, ${N_\uptau}$, and ${N_\upnu}$ are sufficiently large, ${\cal S}_\upvartheta$, ${\cal S}_\uptau$, and ${\cal S}_\upnu$ can approximate the continuous ${\cal B}_\upvartheta ^{}$, ${\cal B}_\uptau$ and ${\cal B}_\upnu$, respectively. Each point in ${\cal S}_\upvartheta$, ${\cal S}_\uptau$, or ${\cal S}_\upnu$ can be interpreted as a physical beam directed toward a specific angular cosine, delay, or Doppler shift. As stated in \cite{mine}, the discrete angle domain ${\cal S}_\upvartheta$, the discrete delay domain ${\cal S}_\uptau$, and the discrete Doppler domain ${\cal S}_\upnu$ are referred to as the spatial beam domain, frequency beam domain, and time beam domain, respectively. Collectively, these three domains form the triple-beam (TB) domain.
					
					To analyze the channel characteristics in the TB domain, we divide the ranges ${\cal B}_\upvartheta ^{}$, ${\cal B}_\uptau$ and ${\cal B}_\upnu ^{}$ into ${\cal B}_\upvartheta ^{} $$=$$ \bigcup\limits_{n = 1}^{{N_\upvartheta }} {{\cal B}_\upvartheta ^n} $, ${\cal B}_\uptau ^{}$${\rm{ = }}$$\bigcup\limits_{n = 1}^{{N_\uptau }} {{\cal B}_\uptau ^n} $ and  ${{\cal B}_\upnu } $$=$$ \bigcup\limits_{n = 1}^{{N_\upnu }} {{\cal B}_\upnu ^n} $, where ${\cal B}_\upvartheta ^n{\rm{ = }}\left[ {\frac{{n - 1 - 0.5{N_\upvartheta }}}{{0.5{N_\upvartheta }}},\frac{{n - 0.5{N_\upvartheta }}}{{0.5{N_\upvartheta }}}} \right)$, ${\cal B}_\uptau ^n = \left[ {\frac{{n - 1}}{{{N_\uptau }\Delta f}},\frac{n}{{{N_\uptau }\Delta f}}} \right)$ and ${\cal B}_\upnu ^n = \left[ {\frac{{n - 1 - 0.5{N_\upnu }}}{{{N_\upnu }{N_{\rm{b}}}{N_{{\rm{sym}}}}}},\frac{{n - 0.5{N_\upnu }}}{{{N_\upnu }{N_{\rm{b}}}{N_{{\rm{sym}}}}}}} \right)$. When ${{N_\upvartheta }}$, ${{N_\uptau}}$ and ${{N_\upnu}}$ are sufficiently large, the subranges ${\cal B}_\upvartheta ^n$, ${\cal B}_\uptau ^n$ and ${\cal B}_\upnu ^n$ will also be sufficiently narrow. In this case, any $\vartheta  \in {\cal B}_\upvartheta ^n$, $\tau  \in {\cal B}_\uptau ^n$ and $\nu  \in {\cal B}_\upnu ^n$ can be approximated by $\vartheta _n = \frac{{n - 1 - 0.5{N_\upvartheta }}}{{0.5{N_\upvartheta }}}$, $\tau_n = \frac{{n - 1}}{{{N_\uptau }\Delta f}}$ and $\nu_n = \frac{{n - 1 - 0.5{N_\upnu }}}{{{N_\upnu }{N_{\rm{b}}}{N_{{\rm{sym}}}}}}$, respectively. Let ${\cal B}_{{\rm{TB}}}^{{n_\upvartheta },{n_\uptau },{n_\upnu }} = {\cal B}_\upvartheta ^{{n_\upvartheta }} \times {\cal B}_\uptau ^{{n_\uptau }} \times {\cal B}_\upnu ^{{n_\upnu }}$ and 
					\begin{equation}
						\!\!\!\!	{\alpha _{{\bf{s}},{n_\upvartheta },{n_\uptau },{n_\upnu }}} \!\!= \sum\limits_{\scriptstyle\left( {p,q} \right):\{ {\vartheta _{{\bf{s}},{l_{p,q}}}},{\tau _{{\bf{s}},{l_{p,q}}}},{\nu _{{\bf{s}},{l_{p,q}}}}\}  \in {\cal B}_{{\rm{TB}}}^{{n_\upvartheta },{n_\uptau },{n_\upnu }}\hfill\atop
							\scriptstyle1 \le p \le P({\bf{s}}),1 \le q \le {L_p}({\bf{s}})\hfill} {\alpha _{{\bf{s}},{l_{p,q}}}^{}},
					\end{equation}
					to represent the sum of the gains of all MPCs located within ${\cal B}_{{\rm{TB}}}^{{n_\upvartheta },{n_\uptau },{n_\upnu }}$ along the TB domain. Then ${\alpha _{{\bf{s}},{n_\upvartheta },{n_\uptau },{n_\upnu }}}$ serves as the channel gain corresponding to the $({n_\upvartheta },{n_\uptau },{n_\upnu })$-th beam in the TB domain. Consequently, the CTF in (\ref{CTF5}) can be expressed with ${\alpha _{{{\bf{s}}},{n_\upvartheta },{n_\uptau },{n_\upnu }}}$ as
					\begin{align}\label{CTF6}
						&{h_a}\left( {n{T_{{\rm{sym}}}},k\Delta f} \right) = \sum\limits_{{n_\upvartheta } = 1}^{{N_\upvartheta }} {\sum\limits_{{n_\uptau } = 1}^{{N_\uptau }} {\sum\limits_{{n_\upnu } = 1}^{{N_\upnu }} {{\alpha _{{{\bf{s}}},{n_\upvartheta },{n_\uptau },{n_\upnu }}}} } } \notag\\
						&\cdot {e^{\bar \jmath 2\pi {\nu _{{n_\upnu }}}n{T_{{\rm{sym}}}}}} \cdot {e^{ - \bar \jmath 2\pi \left( {{f_{\rm{c}}} + k\Delta f} \right)(a-1){\frac{{{d_{{\rm{BS}}}}}}{c}}{\vartheta _{{n_\upvartheta }}}}} \cdot {e^{-\bar \jmath 2\pi k\Delta f{\tau _{{n_\uptau }}}}}.
					\end{align}

					According to the system configuration and frame structure, we can define the following matrices ${{\bf{V}}^{\rm{f}}} \in {{\mathbb{C}}^{K \times {N_\uptau }}}$ and ${\bf{V}}_{\tilde n}^{\rm{t}} \in {{\mathbb{C}}^{{N_{\rm{p}}} \times {N_\upnu }}}$ where ${\left[ {{{\bf{V}}^{\rm{f}}}} \right]_{k,{n_\uptau }}} = {e^{ - \bar \jmath 2\pi k\Delta f{\tau _{{n_\uptau }}}}}$ and ${\left[ {{\bf{V}}_{\tilde n}^{\rm{t}}} \right]_{n,{n_\upnu }}} = {e^{\bar \jmath 2\pi {\nu _{{n_\upnu }}}{t_{{\tilde n},n}}{N_{\rm{b}}}{T_{{\rm{sym}}}}}}$, and the tensor ${{\mathbfcal V}^{\rm{s}}} \in {{\mathbb{C}}^{A \times K \times {N_\upvartheta } \times K}}$, where ${\left[ {{{\mathbfcal V}^{\rm{s}}}} \right]_{a,k,{n_\upvartheta },k}} = {{e^{ - \bar \jmath 2\pi \left( {{f_{\rm{c}}} + k\Delta f} \right)(a-1)\frac{{{d_{{\rm{BS}}}}}}{c}{\vartheta _{{n_\upvartheta }}}}}}$. From (\ref{CTF6}), ${\mathbfcal H}_{{\bf{s}},\tilde n}^{{\rm{SFT}}}$ can be further expressed in terms of ${{\mathbfcal V}^{\rm{s}}}$, ${{{\bf{V}}^{\rm{f}}}}$ and ${{\bf{V}}_{\tilde n}^{\rm{t}}}$ by
					\begin{equation}\label{tbchannel}
						{\mathbfcal H}_{{\bf{s}},\tilde n}^{{\rm{SFT}}} = {{\mathbfcal V}^{\rm{s}}}{*_2}\left( {{{\bf{V}}^{\rm{f}}}{ \times _2}\left( {{\bf{V}}_{\tilde n}^{\rm{t}}{ \times _3}{\mathbfcal H}_{\bf{s}}^{{\rm{TB}}}} \right)} \right),
					\end{equation}
					where ${\mathbfcal H}_{\bf{s}}^{{\rm{TB}}} \in {{\mathbb{C}}^{{N_\upvartheta } \times {N_\uptau } \times {N_\upnu }}}$ represents the TB domain channel tensor at the location ${\bf{s}}$ satisfying ${\left[ {{\mathbfcal H}_{\bf{s}}^{{\rm{TB}}}} \right]_{{n_\upvartheta },{n_\uptau },{n_\upnu }}} = {\alpha _{{\bf{s}},{n_\upvartheta },{n_\uptau },{n_\upnu }}}$.  It should be noted that, when the spatial-wideband effect can be ignored, ${{\mathbfcal V}^{\rm{s}}}$ reduces to ${{\bf{V}}^{\rm{s}}} \in {{\mathbb{C}}^{A \times {N_\upvartheta }}}$ with ${\left[ {{{\bf{V}}^{\rm{s}}}} \right]_{a,{n_\vartheta }}} = {e^{ - \bar \jmath 2\pi (a - 1){d_{{\rm{BS}}}}{\vartheta _{{n_\upvartheta }}}/{\lambda _{\rm{c}}}}}$, and 
					(\ref{tbchannel}) simplifies to \cite{9910031}
					\begin{equation}\label{tbchannel1}
						{\mathbfcal H}_{{\bf{s}},\tilde n}^{{\rm{SFT}}} = {{\bf{V}}^{\rm{s}}}{ \times _1}\left( {{{\bf{V}}^{\rm{f}}}{ \times _2}\left( {{\bf{V}}_{\tilde n}^{\rm{t}}{ \times _3}{\mathbfcal H}_{\bf{s}}^{{\rm{TB}}}} \right)} \right).
					\end{equation}

					To facilitate the subsequent derivation, we adopt the model in (\ref{tbchannel1}) for analysis, while a similar approach also applies to (\ref{tbchannel}). Since both (\ref{tbchannel}) and (\ref{tbchannel1}) establish the relationship between the TB domain and the SFT domain channels, we collectively refer to them as the triple BBCM. In addition to the proposed triple BBCM, 1D BBCMs \cite{10993474} and 2D BBCMs (also known as double BBCMs) \cite{10146318} can be obtained by quantizing the angle domain or the angle-delay domain, respectively. Unlike GBSMs, BBCMs are statistical models that approximate continuous physical domains using discrete beam representations. This modeling approach is increasingly favored in channel acquisition and transmission design \cite{lu2020robust}.

					\subsection{TB Domain sCSI}
					For massive MIMO-OFDM systems, due to the limited number of clusters in the propagation environment, ${{\mathbfcal H}_{\bf{s}}^{{\rm{TB}}}}$ tends to be sparse, and sparsity exists along its triple dimensions \cite{ mine}. Compared with  the SFT domain channel tensor ${\mathbfcal H}_{{\bf{s}},\tilde n}^{{\rm{SFT}}}$, the number of dominate elements in ${{\mathbfcal H}_{\bf{s}}^{{\rm{TB}}}}$ is significantly reduced. As for sCSI, 
					assuming the channel coefficients across different beams follow independent circular symmetric complex Gaussian distributions with zero mean and varying variances, we can define the following tensor \cite{7332961}
					\begin{equation}\label{Wu}
						{\mathbfcal W}_{\bf{s}} = {\mathbb{E}}\left\{ {{\mathbfcal H}_{\bf{s}}^{{\rm{TB}}} \odot {{\left( {{\mathbfcal H}_{\bf{s}}^{{\rm{TB}}}} \right)}^*}} \right\}. 
					\end{equation}
					Then, ${\mathbfcal W}_{\bf{s}}$ is the TB domain power distribution tensor, and the $\left( {{n_\upvartheta },{n_\uptau },{n_\upnu }} \right)$-th element in ${\mathbfcal W}_{\bf{s}}$ corresponds to the power of the $\left( {{n_\upvartheta },{n_\uptau },{n_\upnu }} \right)$-th beam in the TB domain. Define the SFT domain channel covariance matrix as ${\bf{R}}_{{\bf{s}},\tilde n}^{{\rm{SFT}}} = {\mathbb{E}}\left\{ {{\rm{vec}}({\mathbfcal{H}}_{{\bf{s}},\tilde n}^{{\rm{SFT}}}){{\left( {{\rm{vec}}({\mathbfcal{H}}_{{\bf{s}},\tilde n}^{{\rm{SFT}}})} \right)}^{\rm{H}}}} \right\}$. With the channel model in (\ref{tbchannel1}), ${\bf{R}}_{{\bf{s}},\tilde n}^{{\rm{SFT}}}$ can also be expressed with ${\mathbfcal W}_{\bf{s}}$ as \cite{mine,10146318}
					\begin{equation}\label{Rsft1}
						{\bf{R}}_{{\bf{s}},\tilde n}^{{\rm{SFT}}} = {\bf{V}}{\rm{Diag}}({\rm{vec}}({{\mathbfcal W}_{\bf{s}}})){{\bf{V}}^{\rm{H}}},
					\end{equation}
					where ${\bf{V}} \buildrel \Delta \over = {\bf{V}}_{\tilde n}^{\rm{t}} \otimes {{\bf{V}}^{\rm{f}}} \otimes {{\bf{V}}^{\rm{s}}}$.

					Given its sparsity and slow time-varying characteristics, ${\mathbfcal W}_{\bf{s}}$ can be utilized in various applications, including pilot design \cite{mine}, channel estimation \cite{9910031}, and precoding  \cite{xie}. By definition, ${\mathbfcal W}_{\bf{s}}$ is related to the propagation environment between the BS and ${\bf{s}}$, as well as the mobility of the UT. This enables the generation of ${\mathbfcal W}_{\bf{s}}$ using appropriate methods when the PCF at ${\bf{s}}$ and the UT mobility are given, as shown in Fig.~\ref{frame1}. To this end, we aim to propose a method to acquire ${\mathbfcal W}_{\bf{s}}$ directly from the PCF, thereby eliminating the probing overhead associated with traditional methods for acquiring ${\mathbfcal W}_{\bf{s}}$. 
				
					\begin{remark}
						\emph{In existing literature, both the spatial beam domain energy distribution vector (1D beam domain sCSI) and the spatial-frequency (SF) beam domain energy distribution matrix (SF beam domain sCSI) can serve as CFs \cite{10287775, 10146318, 10292876, 11247875}. However, ${\mathbfcal W}_{\bf{s}}$ cannot directly act as a CF of ${\bf{s}}$, as its distribution along the time beam domain (i.e., discrete Doppler domain) depends not only on the location but also on the UT mobility.}
					\end{remark}

					\subsection{Acquiring ${\mathbfcal W}_{\bf{s}}^{}$ from the PCF}
					For any ${\bf{s}} \in {\cal S}$, once the PCF ${{\cal M}_{\bf{s}}}$ is available, the AAD domain power distribution ${S_{{{\bf{s}}}}}\left({{\bf{\Theta }}^\alpha } \right)$ can then be reconstruncted using (\ref{St}). Since the SFT domain channel ${\mathbfcal H}_{{\bf{s}},\tilde n}^{{\rm{SFT}}}$ is strongly influenced by the MPC parameters, as shown in (\ref{CTF5}), the corresponding covariance matrix ${\bf{R}}_{{\bf{s}},\tilde n}^{{\rm{SFT}}}$ is also closely related to the statistics of the MPC parameters, i.e., $S_{\bf{s}}\left( {\bf{\Theta }}^\alpha \right)$. According to \cite{you2015pilot}, we adopt the common uncorrelated scattering approximation where channel gains at different directions are statistically uncorrelated. Thus, for parameter vectors ${{\bf{\Theta }}^\alpha } = {[\beta ,\theta ,\tau ]^{\rm{T}}}$ and ${{{\bf{\bar \Theta }}}^\alpha } = {[\bar \beta ,\bar \theta ,\bar \tau ]^{\rm{T}}}$, ${\mathbb{E}}$$\left\{ {{g_{\bf{s}}}\left( {{{\bf{\Theta }}^\alpha }} \right)g_{\bf{s}}^*\left( {{\bar{\bf{\Theta }}^\alpha }} \right)} \right\}$ $ = $ ${S_{\bf{s}}}\left( {{{\bf{\Theta }}^\alpha }} \right)$ $\delta \left( {\beta  - \bar \beta } \right)$ $\delta \left( {\theta  - \bar \theta } \right)$ $\delta \left( {\tau  - \bar \tau } \right)$ holds. Consequencely, the covariance matrix ${\bf{R}}_{{\bf{s}},\tilde n}^{{\rm{SFT}}}$ can be completely charaterized by ${S_{\bf{s}}}\left( {{{\bf{\Theta }}^\alpha }} \right)$ as
					\begin{align}\label{Rsft2}
						\!{\bf{R}}_{{\bf{s}},\tilde n}^{{\rm{SFT}}}\!\! = \iiint &{{S_{\bf{s}}}\left( {{{\bf{\Theta }}^\alpha }} \right){\bf{v}}_{\tilde n}^{{\rm{sft}}}(\beta ,\theta ,\tau )}{\left[ {\bf{v}}_{\tilde n}^{{\rm{sft}}}(\beta ,\theta ,\tau ) \right]^{\rm{H}}}d\theta d\beta d\tau.
					\end{align}

					

					Equations (\ref{Rsft1}) and (\ref{Rsft2}) establish the relationships between ${\bf{R}}_{{\bf{s}},\tilde n}^{{\rm{SFT}}}$ and the beam domain sCSI, as well as between ${\bf{R}}_{{\bf{s}},\tilde n}^{{\rm{SFT}}}$ and ${S_{\bf{s}}\left( {\bf{\Theta }}^\alpha \right)}$, respectively. Leveraging these relationships, the process of generating sCSI from the PCF can be carried out in the following steps: first, construct ${S_{\bf{s}}\left( {\bf{\Theta }}^\alpha \right)}$ based on the PCF ${{\cal M}_{\bf{s}}}$; then, derive the corresponding covariance matrix ${\bf{R}}_{{\bf{s}},\tilde n}^{{\rm{SFT}}}$ with ${S_{\bf{s}}\left( {\bf{\Theta }}^\alpha \right)}$ by (\ref{Rsft2}). Finally, the problem of obtaining ${{\mathbfcal W}_{\bf{s}}}$ is reformulated as an optimization task, where the objective is to find a non-negative real-valued tensor ${\mathbfcal W}$ such that
					\begin{equation}\label{condition1}
						{\bf{R}}_{{\bf{s}},\tilde n}^{{\rm{SFT}}} = {\bf{V}}{\rm{Diag}}({\rm{vec}}({\mathbfcal W})){{\bf{V}}^{\rm{H}}}. 
					\end{equation}
					This problem is essentially a sparse signal recovery problem. However, commonly used sparse signal reconstruction algorithms are typically designed for recovering complex-valued signals \cite{10146318}, and the condition of non-negative real values is challenging to satisfy. To address this, we propose a novel method for solving ${\mathbfcal W}$.


					
					Due to its non-negativity, the power distribution tensor can be written as the element-wise product of two real tensors, i.e., ${\mathbfcal W} = {\mathbfcal Q} \odot {\mathbfcal Q}$, where ${\mathbfcal Q} \in {{\mathbb{R}}^{{N_\upvartheta } \times {N_\uptau } \times {N_\upnu }}}$. Thus, the problem of obtaining ${\mathbfcal W}$ from the PCF can be transformed into obtaining the real tensor ${\mathbfcal Q}$, thereby avoiding the non-negativity constraint of ${\mathbfcal W}$. On the other hand, when ${\bf{R}}_{{\bf{s}},\tilde n}^{{\rm{SFT}}}$ is generated through the PCF, a natural idea to transform it into the TB domain is operating ${{\rm{diag}}\left\{ {{{\bf{V}}^{\rm{H}}}{\bf{R}}_{{\bf{s}},\tilde n}^{{\rm{SFT}}}{\bf{V}}} \right\}}$ \cite{10146318}. Therefore, onces (\ref{condition1}) holds, ${\mathbfcal Q}$ should also fulfill the following condition:
					\begin{equation}\label{lr}
						\scalebox{0.95}{$	\! {{\rm{diag}}\left\{ {{{\bf{V}}^{\rm{H}}}{\bf{V}}{\rm{Diag}}\left( {{\rm{vec}}({\mathbfcal Q} \odot {\mathbfcal Q})} \right){{\bf{V}}^{\rm{H}}}{\bf{V}}} \right\}} = {{\rm{diag}}\left\{ {{{\bf{V}}^{\rm{H}}}{\bf{R}}_{{\bf{s}},\tilde n}^{{\rm{SFT}}}{\bf{V}}} \right\}},$}
					\end{equation}
					where both sides are real vectors of size ${N_\upvartheta }{N_\uptau }{N_\upnu } \times 1$.

					\begin{figure*}[t!]
						\begin{align}\label{LQ}
							L\left( {\mathbfcal Q} \right) = &\scalebox{1}{$\sum\limits_{{n_\upvartheta } = 1}^{{N_\upvartheta } } {\sum\limits_{{n_\uptau } = 1}^{{N_\uptau } } {\sum\limits_{{n_\upnu } = 1}^{{N_\upnu } } {{{[\bar {\mathbfcal A}]}_{{n_\upvartheta },{n_\uptau },{n_\upnu }}}\log \frac{{{{[\bar {\mathbfcal A}]}_{{n_\upvartheta },{n_\uptau },{n_\upnu }}}}}{{{{[{\mathbfcal A}\left( {\mathbfcal Q} \right)]}_{{n_\upvartheta },{n_\uptau },{n_\upnu }}}}}} } }+ \sum\limits_{{n_\upvartheta } = 1}^{{N_\upvartheta } } {\sum\limits_{{n_\uptau } = 1}^{{N_\uptau } } {\sum\limits_{{n_\upnu } = 1}^{{N_\upnu } } {{{[{\mathbfcal A}\left( {\mathbfcal Q} \right)]}_{{n_\upvartheta },{n_\uptau },{n_\upnu }}}} } } -\sum\limits_{{n_\upvartheta } = 1}^{{N_\upvartheta } } {\sum\limits_{{n_\uptau } = 1}^{{N_\uptau } } {\sum\limits_{{n_\upnu } = 1}^{{N_\upnu } } {{{[\bar {\mathbfcal A}]}_{{n_\upvartheta },{n_\uptau },{n_\upnu }}}} } }.$}	
						\end{align}
						\hrulefill
					\end{figure*}	
					
					To further simplify the computation on both sides of (\ref{lr}), we reshape them into ${{N_\upvartheta} \times {N_\uptau} \times {N_\upnu}}$ tensors, denoted as ${\mathbfcal A}\left( {\mathbfcal Q} \right)$ and ${\bar {\mathbfcal A}}$, respectively. To quantify the similarity between ${\mathbfcal A}\left( {\mathbfcal Q} \right)$ and ${\bar {\mathbfcal A}}$, we adopt the Kullback-Leibler (KL) divergence as the metric \cite{9910031}. According to its definition, the KL divergence between ${\mathbfcal A}\left( {\mathbfcal Q} \right)$ and ${\bar {\mathbfcal A}}$ is given by (\ref{LQ}) at the top of the next page. A smaller KL divergence indicates a closer match between ${\mathbfcal A}\left( {\mathbfcal Q} \right)$ and ${\bar {\mathbfcal A}}$. Based on this, the problem can be formulated as \begin{equation}\label{finalp} 
						{{\mathbfcal Q}^*} = \mathop {\arg \min }\limits_{\mathbfcal Q} L\left( {\mathbfcal Q} \right). 
					\end{equation}

					In practice, during the algorithm execution, ${\mathbfcal A}\left( {\mathbfcal Q} \right)$ can also be directly computed in the following manner:
					\begin{equation}\label{AQ}
						{\mathbfcal A}\left( {\mathbfcal Q} \right) = {{\bf{T}}^{\rm{s}}}{ \times _1}\left( {{{\bf{T}}^{\rm{f}}}{ \times _2}\left( {{{\bf{T}}^{\rm{t}}}{ \times _3}({\mathbfcal Q} \odot {\mathbfcal Q})} \right)} \right),
					\end{equation}
					where 
					\begin{subequations}\label{T1}
						\begin{align}
							&\scalebox{0.99}{${{\bf{T}}^{\rm{s}}} = \left( {{{({{\bf{V}}^{\rm{s}}})}^{\rm{H}}}{{\bf{V}}^{\rm{s}}}} \right) \odot {\left( {{{({{\bf{V}}^{\rm{s}}})}^{\rm{H}}}{{\bf{V}}^{\rm{s}}}} \right)^*}$},\label{T11}\\
							&\scalebox{0.99}{${{\bf{T}}^{\rm{f}}} = \left( {{{({{\bf{V}}^{\rm{f}}})}^{\rm{H}}}{{\bf{V}}^{\rm{f}}}{\rm{ }}} \right) \odot {\left( {{{({{\bf{V}}^{\rm{f}}})}^{\rm{H}}}{{\bf{V}}^{\rm{f}}}{\rm{ }}} \right)^*}{\rm{ }}$},\label{T12}\\
							&\scalebox{0.99}{${{\bf{T}}^{\rm{t}}} = \left( {{{({\bf{V}}_{\tilde n}^{\rm{t}})}^{\rm{H}}}{\bf{V}}_{\tilde n}^{\rm{t}}} \right) \odot {\left( {{{({\bf{V}}_{\tilde n}^{\rm{t}})}^{\rm{H}}}{\bf{V}}_{\tilde n}^{\rm{t}}} \right)^*}$}.\label{T13}
						\end{align}
					\end{subequations}
					By executing (\ref{AQ})-(\ref{T1}) to calculate ${\mathbfcal A}\left( {\mathbfcal Q} \right)$, we bypass the need for diagonalization and reshaping operations.

					For the optimization problem in (\ref{finalp}), we first present the following theorem to facilitate the subsequent solution.
					
					\begin{theorem}\label{theorem1}
						The gradient of $L\left( {\mathbfcal Q} \right)$ is given by
						\begin{equation}\label{result}
							\scalebox{1.1}{$\frac{{\partial L\left( {\mathbfcal Q} \right)}}{{\partial {\mathbfcal Q}}}$} = 2{\mathbfcal Q} \odot \left({{\bf{T}}^{\rm{s}}}{ \times _1}\left( {{{\bf{T}}^{\rm{f}}}{ \times _2}\left( {{{\bf{T}}^{\rm{t}}}{ \times _3}{\mathbfcal C}} \right)} \right)\right),
						\end{equation}
						where ${\mathbfcal C} \in {{\mathbb{R}}^{{N_\upvartheta } \times {N_\uptau } \times {N_\upnu }}}$ satisfying
						\begin{equation}
							{\left[ {\mathbfcal C} \right]_{{n_\upvartheta },{n_\uptau },{n_\upnu }}} = 1 - \scalebox{1.1}{$\frac{{{{[\bar {\mathbfcal A}]}_{{n_\upvartheta },{n_\uptau },{n_\upnu }}}}}{{{{[{\mathbfcal A}\left( {\mathbfcal Q} \right)]}_{{n_\upvartheta },{n_\uptau },{n_\upnu }}}}}$}.
						\end{equation}
					\end{theorem}
					\emph{Proof}: See Appendix \ref{app1}.
					
					
					With the gradient of $L\left( {\mathbfcal Q} \right)$ given in (\ref{result}), we can employ the gradient descent method to update ${\mathbfcal Q}$ in the $d$th iteration as follows:
					\begin{equation}\label{iter}
						{{\mathbfcal Q}^{(d + 1)}} = {{\mathbfcal Q}^{(d)}} - \scalebox{1.1}{${\delta ^{(d)}}\frac{{\partial L\left( {{{\mathbfcal Q}^{(d)}}} \right)}}{{\partial {{\mathbfcal Q}^{(d)}}}}$}
					\end{equation}
					where ${\delta ^{(d)}}$ is the step size. The detailed sCSI acquisition method is presented in Algorithm \ref{SCFgen}. It can be observed that the proposed Algorithm \ref{SCFgen} relies solely on tensor product operations, thereby avoiding computationally intensive procedures such as matrix inversion.

					\begin{algorithm}[t!]
						\caption{sCSI Acquisition}
						\label{SCFgen}
						\begin{algorithmic}[1]
							\footnotesize
							\Require The PCF ${{\cal M}_{\bf{s}}}$, the beam matrices ${{{\bf{V}}^{\rm{s}}}}$, ${{{\bf{V}}^{\rm{f}}}}$ and ${{\bf{V}}_{\tilde n}^{\rm{t}}}$, the mobility of the UT; 
							\Ensure The sCSI ${{\mathbfcal W}_{\bf{s}}}$;
							\State Generating ${\bf{R}}_{{\bf{s}},\tilde n}^{{\rm{SFT}}}$ with the PCF;
							\State Select the maximum iteration number $d_{\rm {max}}$,  the minimum step size ${\delta _{\min }}$, and the parameter $\rho  \in (0,1]$ to adjust the step size;
							\State  Initialization: $d = 0$,  a suitable value for ${\delta^{(0)}}$, ${[{{\mathbfcal Q}^{(0)}}]_{{n_\upvartheta },{n_\uptau },{n_\upnu }}} = \sqrt {{{\left[ {\bar {\mathbfcal A}/({N_\upvartheta }{N_\uptau }{N_\upnu })} \right]}_{{n_\upvartheta },{n_\uptau },{n_\upnu }}}} $ for all the elements in ${{\mathbfcal Q}^{(0)}}$;
							\Repeat
							\State Calculate $\frac{{\partial L\left( {\mathbfcal Q} \right)}}{{\partial {\mathbfcal Q}}}$ by (\ref{result});
							\State Update ${{\mathbfcal Q}^{(d + 1)}}$ by (\ref{iter});
							\If{$L\left( {{{\mathbfcal Q}^{(d + 1)}}} \right) > L\left( {{{\mathbfcal Q}^{(d)}}} \right)$}
							\State ${\delta ^{(d)}} = \rho {\delta ^{(d)}}$, ${{\mathbfcal Q}^{(d + 1)}} = {{\mathbfcal Q}^{(d)}}$;
							\Else  
							\State break;
							\EndIf
							\State $d = d + 1$;
							\Until{$d = {d_{\max }}$ or ${\delta ^{(d)}} \le {\delta _{\min }}$ }
							\State \Return ${{\mathbfcal W}_{\bf{s}}} = {{\mathbfcal Q}^{(d)}} \odot {{\mathbfcal Q}^{(d)}}$.
						\end{algorithmic}
					\end{algorithm}

					From the definitions in Section~\ref{section4}-A, the beam matrices ${\bf V}^{\rm{s}}$, ${\bf V}^{\rm{f}}$, and ${\bf V}^{\rm{t}}_{\tilde n}$ can all be derived from the DFT matrix or its partial variants through elementary transformations. Specifically, ${{\bf{V}}^{\rm{s}}} = {\left[ {{{\bf{F}}_{{N_\upvartheta },{N_\upvartheta }/2}}} \right]_{1:A,:}}$, ${{\bf{V}}^{\rm{f}}} = {\left[ {{{\bf{F}}_{{N_\uptau }}}} \right]_{{\cal K},:}}$, and ${\bf{V}}_{\tilde n}^{{\rm{t}}} = {[{{\bf{\Gamma }}_{{N_\upnu },{{t_{\tilde n,1}}}}}{\bf{F}}_{{N_\upnu },{N_\upnu }/2}^*]_{{1}:N_{\rm p},:}}$. Leveraging these structured forms of the beam matrices, the implementation of Algorithm~\ref{SCFgen} can be further simplified. Based on this observation, we now present Theorem 2.

					\begin{theorem}\label{theorem2}
						The matrices ${{{\bf{T}}^{\rm{s}}}}$, ${{{\bf{T}}^{\rm{f}}}}$ and ${{{\bf{T}}^{\rm{t}}}}$ are all circulant matrices, and 
						\begin{subequations}\label{T2}
							\begin{align}
								&\scalebox{0.99}{${{\bf{T}}^{\rm{s}}} = {\bf{F}}_{{N_\upvartheta }}^{\rm{H}}{\rm{Diag}}\{ {{\bf{t}}^{\rm{s}}}\} {{\bf{F}}_{{N_\upvartheta }}}$},\label{T21}\\
								&\scalebox{0.99}{${{\bf{T}}^{\rm{f}}} = {\bf{F}}_{{N_\uptau }}^{\rm{H}}{\rm{Diag}}\{ {{\bf{t}}^{\rm{f}}}\} {{\bf{F}}_{{N_\uptau }}}$},\label{T22}\\
								&\scalebox{0.99}{${{\bf{T}}^{\rm{t}}} = {\bf{F}}_{{N_\upnu }}^{\rm{H}}{\rm{Diag}}\{ {{\bf{t}}^{\rm{t}}}\} {{\bf{F}}_{{N_\upnu }}}$},\label{T23}
							\end{align}
						\end{subequations}
						where
						\begin{subequations}\label{T4}
							\begin{align}
								&\scalebox{0.99}{${{\bf{t}}^{\rm{s}}} = \frac{1}{{{N_\upvartheta }}}{{\bf{F}}_{{N_\upvartheta }}}\left[ {\left( {{\bf{F}}_{{N_\upvartheta }}^{\rm{H}}{{\bf{p}}^{\rm{s}}}} \right) \odot {{\left( {{\bf{F}}_{{N_\upvartheta }}^{\rm{H}}{{\bf{p}}^{\rm{s}}}} \right)}^*}} \right]$},\label{T31}\\
								&\scalebox{0.99}{${{\bf{t}}^{\rm{f}}} = \frac{1}{{{N_\uptau }}}{{\bf{F}}_{{N_\uptau }}}\left[ {\left( {{\bf{F}}_{{N_\uptau }}^{\rm{H}}{{\bf{p}}^{\rm{f}}}} \right) \odot {{\left( {{\bf{F}}_{{N_\uptau }}^{\rm{H}}{{\bf{p}}^{\rm{f}}}} \right)}^*}} \right]$},\label{T32}\\
								&\scalebox{0.99}{${{\bf{t}}^{\rm{t}}} = \frac{1}{{{N_\upnu }}}{{\bf{F}}_{{N_\upnu }}}\left[ {\left( {{\bf{F}}_{{N_\upnu }}^{\rm{H}}{{\bf{p}}^{\rm{t}}}} \right) \odot {{\left( {{\bf{F}}_{{N_\upnu }}^{\rm{H}}{{\bf{p}}^{\rm{t}}}} \right)}^*}} \right]$},\label{T33}
							\end{align}
						\end{subequations}
						and \scalebox{0.96}{${{\bf{p}}^{\rm{s}}} = {[{{\bf{1}}_{1 \times A}},{{\bf{0}}_{1 \times ({N_\upvartheta } - A)}}]^{\rm{T}}}$}, \scalebox{0.96}{${{\bf{p}}^{\rm{f}}} = [{{\bf{0}}_{1 \times ({k_0} - 1)}},{{\bf{1}}_{1 \times K}},$}
						\scalebox{0.92}{${{\bf{0}}_{1 \times ({N_{{\rm{\uptau}}}} - K - {k_0} + 1)}}{]^{\rm{T}}}$},\scalebox{0.96}{${{\bf{p}}^{\rm{t}}} =$}\scalebox{0.92}{$ {{\bf{\Gamma }}_{{N_\upnu },{t_{\tilde n,1}}}}\!\!\!\left( {{{[1,{{\bf{0}}_{1 \times ({N_{{\rm{\upnu}}}} - {N_{\rm{p}}})}},{{\bf{1}}_{1 \times ({N_{\rm{p}}} - 1)}}]}^{\rm{T}}}} \right)$}.
					\end{theorem}
					\emph{Proof}: See Appendix \ref{app2}.		
					
					According to Theorem 2, the $n$-mode product operations in (\ref{AQ}) can be efficiently accelerated using FFTs. Substituting ${{\bf{T}}^{\rm{s}}}$, ${{\bf{T}}^{\rm{f}}}$, and ${{\bf{T}}^{\rm{t}}}$ with the expressions given in (\ref{T21})-(\ref{T23}) reduces their computational complexity to ${\cal O}$$\left( {N_\upvartheta }{N_\uptau }{N_\upnu }{{\log }_2}{N_\upvartheta }\right.$$ +$ $ {N_\upvartheta }N_\uptau ^{}{N_\upnu }{{\log }_2}{N_\uptau } +$$\left. {N_\upvartheta }{N_\uptau }{N_\upnu }{{\log }_2}{N_\upnu } \right)$. This result captures the core online cost of the proposed Algorithm~\ref{SCFgen}. To further analyze its complexity and overhead, we provide a theoretical comparison with the standard online probing approach in \cite{mine}, which is also adopted as our simulation baseline in the next section. Both methods rely on iterative updates; let $d_{\max,\text{probing}}$ and $d_{\max,\text{PCF}}$ denote their respective numbers of iterations required for convergence. The overall online computational complexity of generating $N_{\upvartheta}\times N_{\uptau} \times N_{\upnu}$-dimensional TB domain sCSI for $U$ UTs is summarized in Table~\ref{Comparecom}.

					Although the derivation above uses the triple BBCM in (\ref{tbchannel1}), the same framework can be extended to other BBCMs. For example, for the 2D BBCM in \cite{10292876, 11247875}, a processing flow similar to Algorithm~\ref{SCFgen} can be applied to obtain SF beam domain sCSI. For more complex scenarios, such as near-field propagation, the spatial beam matrix ${\bf V}^{\rm{s}}$ can be replaced by a polar domain beam matrix under near-field conditions \cite{9693928}; regardless of such variations, the core idea and overall workflow of the algorithm remain unchanged. For a complete comparison, Table~\ref{Comparecom} also lists the computational complexity of both methods for generating SF beam domain sCSI. More importantly, the proposed PCF-based method introduces zero online pilot overhead, because all required PCFs are obtained during the offline phase.

					Building on the theoretical comparison in Table~\ref{Comparecom}, we proceed to a concrete numerical analysis in the next section. There, we evaluate the actual computational cost using the specific parameters from our simulation scenarios and the iteration counts observed during algorithm execution.

					\definecolor{lightblue}{rgb}{0.93,0.95,1.0}
					\begin{table*}[t]
						\captionsetup{font=footnotesize}
						\caption{Computational Complexity and Overhead Comparison for Acquiring sCSI}
						\label{Comparecom}
						\centering
						\scriptsize
						\renewcommand{\arraystretch}{1.5}
						\setlength{\tabcolsep}{2pt} 
						\begin{tabular}{>{\raggedright\arraybackslash}p{3cm}
								>{\raggedright\arraybackslash}p{2cm}
								>{\raggedright\arraybackslash}p{5.5cm}
								>{\raggedright\arraybackslash}p{5.5cm}
							}
							\toprule
							\textbf{Method Category} &
							\textbf{Pilot Overhead} &
							\textbf{Computational Complexity for SF Beam Domain sCSI} &
							\textbf{Computational Complexity for TB Domain sCSI}  \\
							\midrule
							\rowcolor{lightblue}
							Online Probing&
							Yes &
							${\cal O}(AK{N_\upvartheta }{N_\uptau }{d_{{\rm{max,probing}}}}U)$
							&
							$\mathcal{O}(AKN_p N_\upvartheta N_\uptau N_\upnu {d_{{\rm{max,probing}}}}U)$  \\
							
							Proposed PCF-Based &
							No &
							$\mathcal{O}((N_\upvartheta N_\uptau \log_2 N_\upvartheta + N_\upvartheta N_\uptau  \log_2 N_\uptau)d_{\max,\mathrm{PCF}}U)$ &
							$\mathcal{O}((N_\upvartheta N_\uptau N_\upnu \log_2 N_\upvartheta + N_\upvartheta N_\uptau N_\upnu \log_2 N_\uptau + N_\upvartheta N_\uptau N_\upnu \log_2 N_\upnu)d_{\max,\mathrm{PCF}}U)$ \\
							\bottomrule
						\end{tabular}
					\end{table*}

						\section{Simulation Results}\label{section5}
						In this section, we present simulation results to validate the effectiveness of the proposed methods. The simulations are conducted using the 6G pervasive channel model (6GPCM) generator \cite{9786750}, which provides AAD domain channel parameters at various locations. These parameters are subsequently converted into the corresponding CIR and CTF. However, the accurate beam domain sCSI is not directly available from the generator.  Prior works on sCSI acquisition \cite{9910031, mine,8074806,10146318,9592715} have demonstrated that long-term online probing is an accurate and practically effective mechanism for acquiring sCSI for multi-UT channel estimation and precoding. Although the implementation details differ across studies, the results in \cite{10146318} show that all the methods yield highly comparable sCSI estimates. Following this insight, and given that \cite{mine} specifically validates the use of probed sCSI for channel estimation, we adopt its online probing algorithm as a representative baseline. To benchmark against these methods, we generate sCSI using both the conventional online probing approach and the proposed PCF-based Algorithm~\ref{SCFgen}. The resulting sCSI is then incorporated into the channel estimation procedure, where it plays two roles: (i) serving as reference information for scheduling pilots across different UTs, and (ii) acting as prior statistical knowledge to approach the minimum mean square error (MMSE) performance. When the estimation performance obtained using PCF-generated sCSI matches that achieved with online probing-generated sCSI, it confirms that the proposed PCF-based approach constitutes a reliable alternative for acquiring the sCSI required for channel estimation. This, in turn, verifies both the accuracy and the practical utility of the proposed method.

						We first generate the PCFs at various locations using the method outlined in Section~\ref{section3}. To distinguish the UT involved in PCF acquisition from those used in subsequent channel estimation applications, we refer to the UT in the PCF acquisition phase as the test terminal (TT). Both the BS and the TT are equipped with multi-antenna arrays to enhance spatial resolution and enable accurate angle estimation \cite{10179246}. As shown in Table~\ref{tablepara}, the TT adopts a $32 \times 1$ ULA during this stage. By moving the TT across different grid regions, channel data are collected at a number of locations, and the corresponding PCFs are extracted. Each PCF, together with its associated location index, is stored to form the cell-specific channel charting with PCFs. Once the offline PCF acquisition is completed, we proceed to evaluate its effectiveness in an online multi-UT channel estimation scenario. A group of single-antenna mobile UTs are introduced within the coverage area. Algorithm~1 is then evaluated in a scenario where the BS simultaneously serves these UTs.

						We conduct simulations in two types of propagation scenarios, namely "${\rm{3GPP\_38}}.{\rm{901\_UMa\_NLOS}}$" and "${\rm{3GPP\_38}}.{\rm{901\_RMa\_NLOS}}$". The propagation environment is divided into $1$m $\times$ $1$m grids for analysis. The detailed parameter configurations are listed in Table \ref{tablepara}. We evaluate the performance of sCSI-assisted channel estimation under varying numbers of UT ($U = 48$ and $U = 180$) and different levels of UT mobility. The initial positions of the UTs are randomly selected within the coverage area, and each UT moves at a speed of ${v_{{\rm{move}}}}$, with a movement direction ${\alpha _{{\rm{move}}}}$ randomly chosen from the range $[0, 2\pi)$. 
						
						\newcolumntype{L}{>{\hspace*{-\tabcolsep}}l}
						\newcolumntype{R}{c<{\hspace*{-\tabcolsep}}}
						\definecolor{lightblue}{rgb}{0.93,0.95,1.0}
						\definecolor{lightgreen}{rgb}{0.95,1.0,0.93}
						\begin{table}[htbp]
							\captionsetup{font=footnotesize}
							\caption{Parameter Settings}\label{tablepara}
							\centering
							\ra{1.3}
							\scriptsize
							\begin{tabular}{LR}
								\toprule
								\textbf{Parameter} &  textbf{Value}\\
								\midrule
								\rowcolor{lightblue}
								Number of BS antennas, $A$&$128$\\
								Number of TT antennas&$32$\\
								\rowcolor{lightblue}
								Number of UT antennas &$1$\\
								Carrier frequency $f_{\rm c}$ &$5.8$ GHz\\
								\rowcolor{lightblue}
								Number of subcarriers, $N_{\rm c}$ &2048\\	 
								CP length $N_{\rm g}$ &144\\
								\rowcolor{lightblue}
								Number of valid subcarriers $K$ &360\\	 
								Subcarrier spacing  $\Delta f$&$15$ kHz\\  
								\rowcolor{lightblue}
								Number of slot in each frame, $N_{\rm p}$ &8\\	                 
								Number of symbols in each slot, $N_{\rm b}$&$14$\\
								\rowcolor{lightblue}
								Speed of UTs, $v_{\rm {move}}$ &3, 15, 25 km$/$h \\	                            
								Number of UTs, $U$&$48 \& 180$\\
								\bottomrule
							\end{tabular}
						\end{table}

						We incorporate the sCSI obtained via both online probing and Algorithm~\ref{SCFgen} into the multi-UT channel estimation process. For channel estimation, we adopt the method proposed in \cite{mine}, which has demonstrated superior performance in multi-UT massive MIMO-OFDM systems. This method requires the SF beam domain or TB domain sCSI as prior information. To evaluate the performance, we use the normalized mean square error (NMSE). During the analysis time, the location of the $u$-th UT is denoted as ${{{\bf{s}}_u}}$, and its SFT domain channel tensor is represented as ${{\mathbfcal H}_{{{\bf{s}}_u},\tilde{n}}^{{\rm{SFT}}}}$. The NMSE for SFT domain channel estimation in the current frame is calculated as follows:
						\begin{equation}
							{\varepsilon _{{\rm{NMSE}}}} = \frac{1}{{{S_{{\rm{sam}}}}U}}\sum\limits_{s = 1}^{{S_{{\rm{sam}}}}} {\sum\limits_{u = 1}^{U} {\frac{{\left\| {{\mathbfcal H}_{{{\bf{s}}_u},\tilde n}^{{\rm{SFT}}}\left( s \right) - \hat {\mathbfcal H}_{{{\bf{s}}_u},\tilde n}^{{\rm{SFT}}}\left( s \right)} \right\|_2^2}}{{\left\| {{\mathbfcal H}_{{{\bf{s}}_u},\tilde n}^{{\rm{SFT}}}\left( s \right)} \right\|_2^2}}} } , 
						\end{equation}
						where ${{S_{{\rm{sam}}}}}$ is the number of channel samples within a frame, ${{\mathbfcal H}_{{{\bf{s}}_u},\tilde n}^{{\rm{SFT}}}\left( s \right)}$ is the $s$-th channel tensor sample and ${\hat {\mathbfcal H}_{{{\bf{s}}_u},\tilde n}^{{\rm{SFT}}}\left( s \right)}$ is the corresponding estimate.

						\subsection{ PCF-based vs. Online Probing for sCSI-aided Channel Estimation
						}
						Fig.~\ref{U40} shows the NMSE performance of channel estimation assisted by different sCSI types in the UMa scenario with $U=48$ UTs under various UT speeds. The beam domain dimension is set as $N_{\upvartheta} = 2A$, $N_{\uptau} = 2K$, $N_{\upnu} = 4N_{\rm p}$, i.e., $N_{\upvartheta} = 256$, $N_{\uptau} =720$, $N_{\upnu} = 32$. The legends in Fig.~\ref{U40} include the following types of sCSI:
						\begin{itemize}
							\item[$\bullet$] \emph{SF Probing}: SF beam domain sCSI obtained through the online probing method;
							\item[$\bullet$] \emph{SF Algorithm 1}:SF beam domain sCSI obtained from PCF using Algorithm~\ref{SCFgen};
							\item[$\bullet$] \emph{TB Probing}: TB domain sCSI obtained through the online probing method;	
							\item[$\bullet$] \emph{TB Algorithm 1}: TB domain sCSI obtained from PCF using Algorithm~\ref{SCFgen}.
						\end{itemize}
						
						Simulation results demonstrate that the sCSI generated by Algorithm~\ref{SCFgen} achieves comparable NMSE performance to that obtained via uplink probing, regardless of whether the SF beam domain or the TB domain sCSI is used. Moreover, across varying levels of UT mobility, both sCSI acquisition methods consistently deliver similar channel estimation performance. Notably, the TB domain sCSI exhibits a stronger capability to capture Doppler domain channel characteristics compared to its SF domain counterpart. As shown in Fig.~\ref{U40}, the TB domain sCSI achieves superior NMSE performance, attributed to its more comprehensive representation of channel properties across multiple domains.

						Then, we compare the NMSE performance of different TB domain sCSI-assisted channel estimation methods under varying values of $N_{\upvartheta}$, $N_{\uptau}$, and $N_{\upnu}$. As shown in Fig.~\ref{NNN}, the channel estimation aided by the sCSI generated from Algorithm 1 consistently achieves NMSE performance comparable to that obtained via online probing, even as the number of beams in the TB domain changes.

						\begin{figure}[t]
							\centering
							\includegraphics[width =190pt]{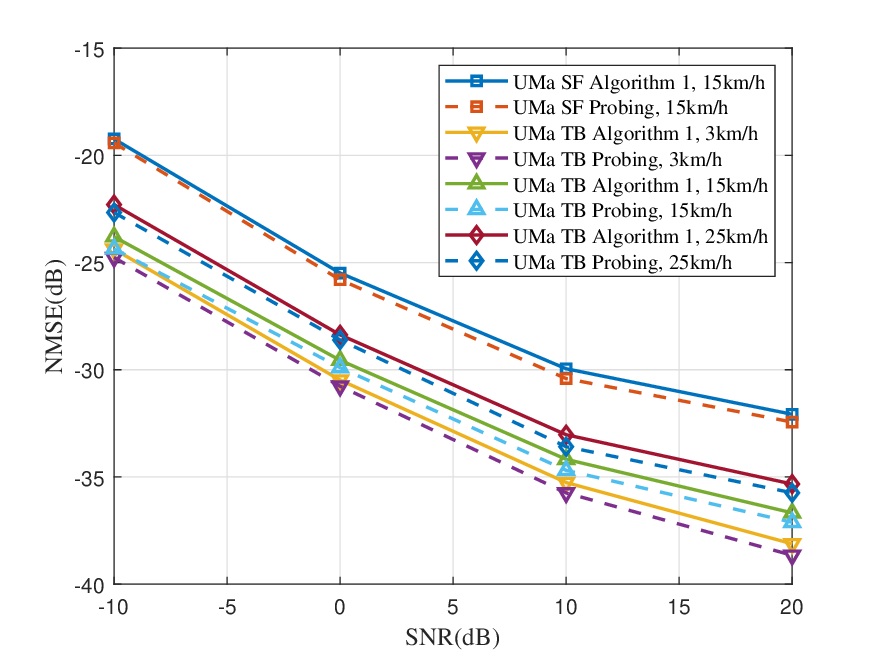}
							\caption{NMSE versus SNR for sCSI acquired by different methods under UMa scenario with varying UT speeds, $U = 48$.}
							\label{U40}
						\end{figure}

						\begin{figure}[t]
							\centering
							\includegraphics[width =190pt]{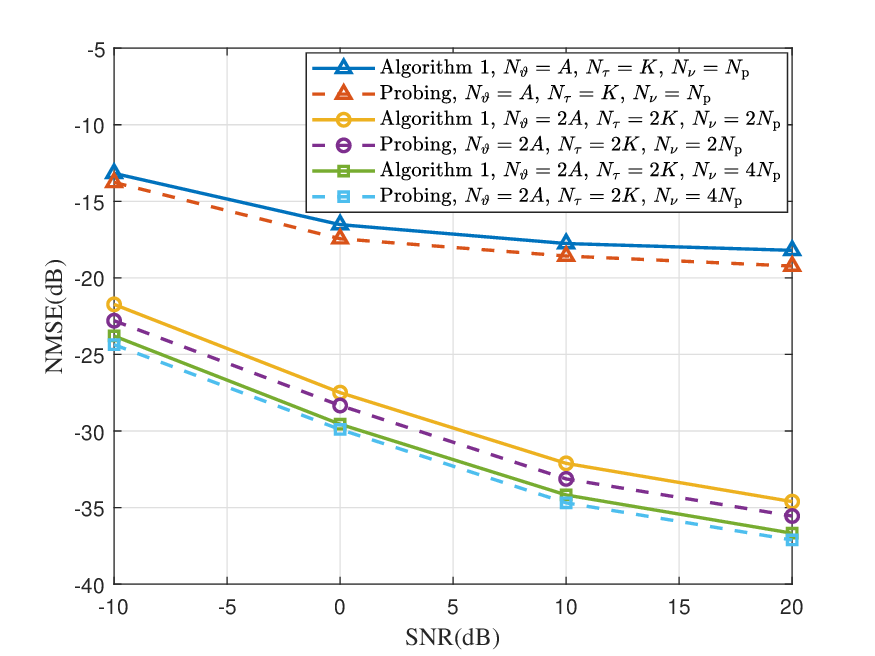}
							\caption{NMSE versus SNR for TB domain sCSI acquired by different methods under UMa scenario with varying $N_{\upvartheta}$, $N_{\uptau}$ and $N_{\upnu}$, $U = 48$, $v_{\rm{move}} = 15$ km$/$h.}
							\label{NNN}
						\end{figure}
						
						As the number of UTs increases, the growing inter-UT interference makes channel estimation more challenging, leading to greater dependence on accurate sCSI \cite{mine}. To investigate this effect, we switch the simulation scenario from Fig.~\ref{U40} to $U=180$ and ${v_{{\rm{move}}}} = 15$ km$/$h with the NMSE performance of various sCSI-assisted channel estimation methods shown in Fig.~\ref{U180}. Along the TB domain, we keep $N_{\upvartheta} = 2A$, $N_{\uptau} = 2K$, $N_{\upnu} = 4N_{\rm p}$. The results reveal that Algorithm~\ref{SCFgen} is still able to approximate the NMSE performance of the online probing method. Moreover, as the number of UTs increases, the NMSE performance gain of the TB domain sCSI-assisted channel estimation method becomes even more pronounced compared to the SF beam domain sCSI-assisted method, showing a greater improvement than in Fig.~\ref{U40}. This is because appropriate UT pilot scheduling along the Doppler domain, assisted by the TB domain sCSI, can further suppress inter-UT pilot interference and lead to a better estimation performance\cite{mine}.
						
						\begin{figure}[t]
							\centering
							\includegraphics[width =190pt]{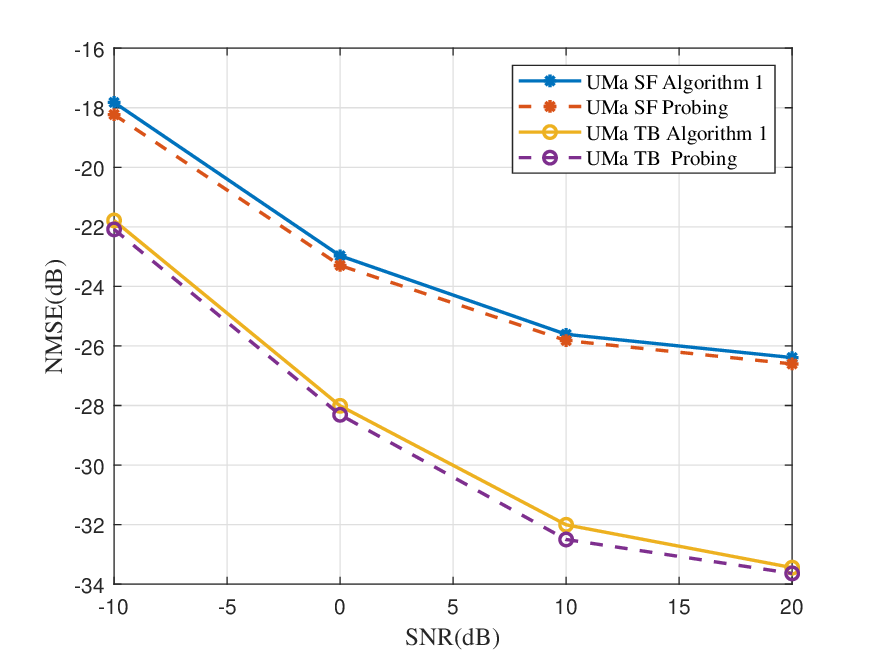}
							\caption{NMSE versus SNR for sCSI acquired by different methods under UMa scenario, $U=180$, $v_{\rm{move}} = 15$ km$/$h.}
							\label{U180}
						\end{figure}

						\begin{figure}[t]
							\centering
							\includegraphics[width =190pt]{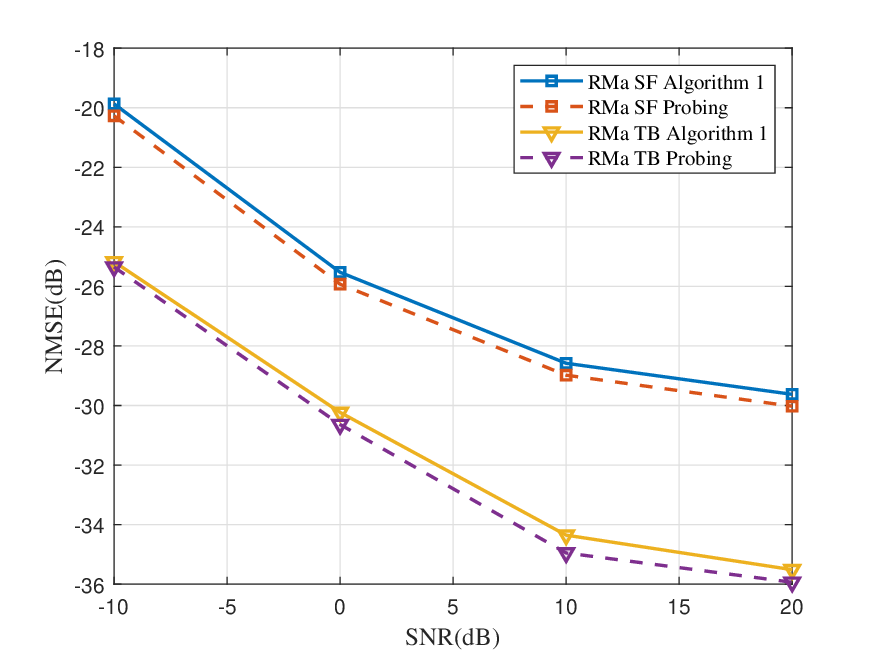}
							\caption{NMSE versus SNR for sCSI acquired by different methods under RMa scenario, $U=180$, $v_{\rm{move}} = 15$ km$/$h.}
							\label{RMA}
						\end{figure}

						Fig.~\ref{RMA} presents the NMSE performance of different sCSI-assisted channel estimation methods for $U=180$ UTs, ${v_{{\rm{move}}}} = 15$ km$/$h in RMa environments. It can be observed that in the RMa scenario, the sCSI generated by both methods still yields comparable NMSE performance in channel estimation. Due to the increased sparsity of the RMa channel, pilot interference among UTs can be more effectively suppressed through appropriate pilot scheduling based on sCSI. As a result, the NMSE performance in this scenario is consistently better than that observed in the corresponding UMa scenario.

						In our simulations, the iterative procedures converge with $d_{\max,\text{probing}}=240$ for online probing and $d_{\max,\text{PCF}}=50$ for the proposed method. Plugging these values and the system dimensions into the complexity orders in Table~\ref{Comparecom} confirms that the actual online computational load of our method is substantially lower than that of online probing, consistent with the theoretical analysis. This efficiency gain, combined with zero online pilot overhead, underscores a practical advantage.

					The storage efficiency of our PCF-based channel charting is another key advantage. While existing works \cite{10146318, 10292876} that directly store the SF beam domain sCSI as CFs require $256 \times 720 = 184,320$ elements per location, our PCF representation, with a maximum of 20 clusters, stores only 140 parameters per location (as defined in (\ref{SCF})). This represents a reduction in storage cost of over three orders of magnitude.
					
					The simulation results further demonstrate the robustness and flexibility inherent to the proposed PCF-based method. In practice, a single set of PCFs, once obtained, can be used to generate sCSI under varying formats and for different target beam domains (e.g., from SF beam to TB domain) without the need for re-probing or retraining. This is different from online probing methods, whose performance is coupled to the specific scenario for which pilots are transmitted, often necessitating a complete re-estimation process when conditions change. The ability of the PCF-based method to maintain consistent sCSI reconstruction quality across these variations is a key advantage validated by our experiments.
					
					\subsection{Extensions: PCF-Enabled sCSI-Assisted channel estimation and Comparison with UL-Based Channel Charting}
					Finally, to further illustrate the effectiveness of the proposed channel charting with PCFs in channel estimation, we include an additional comparison with representative UL-based channel charting methods mentioned in Section~\ref{section1}. Unlike channel charting with CFs, which can reconstruct explicit sCSI and directly support pilot design and MMSE-type estimation, UL-based channel charting cannot generate physical channel parameters or sCSI, and its role in channel estimation is limited to UT grouping and phase-shifted orthogonal pilot (PSOP) reuse \cite{9968109,che2024channel}. In the UMa scenario with $U = 48$ UTs moving at $v_{\rm{move}} = 15$ km/h, we consider the following UL-based CC variants:
					\begin{itemize}
						\item \emph{Angle UL-based CC}: UL-based CC constructed from angular domain CSI, used for UT grouping, PSOP reuse, and subsequent channel estimation;
						\item \emph{Angle-delay UL-based CC}: UL-based CC constructed from angular-delay domain CSI, used for UT grouping, PSOP reuse, and subsequent channel estimation.
					\end{itemize}
					These UL-based methods are implemented following the procedures in \cite{11223663,9968109,che2024channel}. As shown in Fig.~\ref{CCcompare}, while the angle-delay variant improves pilot reuse over the purely angular approach, both remain fundamentally constrained by the lack of sCSI reconstruction. In contrast, PCF-based charting generates accurate, physically meaningful sCSI, providing richer prior information and achieving significantly lower NMSE. These results highlight that, beyond offering an intuitive, physics-driven channel representation, channel charting with PCFs provides a more practical and flexible framework for sCSI-aided channel estimation than UL-based channel charting.


					\begin{figure}[t]
						\centering
						\includegraphics[width =190pt]{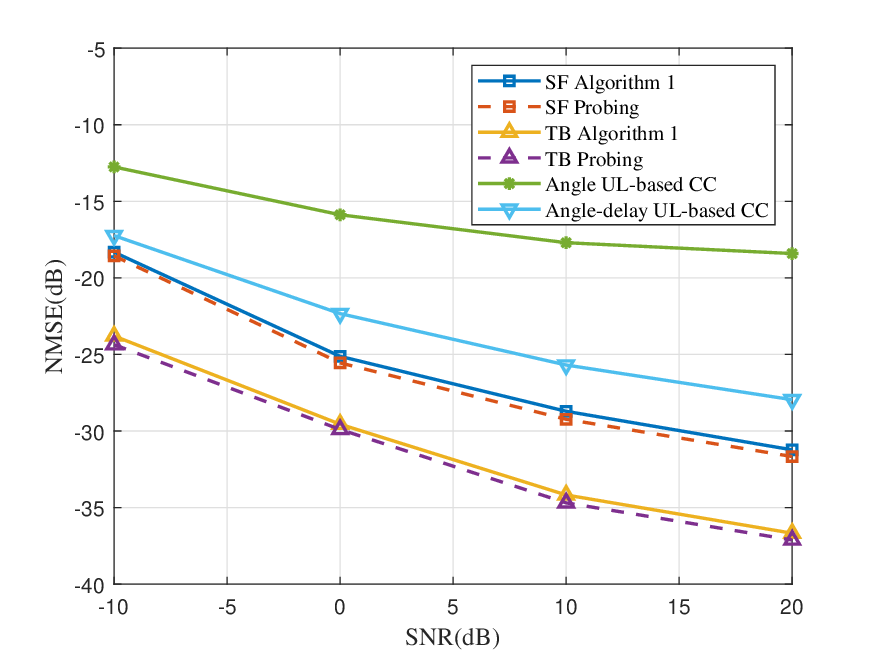}
						\caption{NMSE versus SNR by different methods under UMa scenario with $U = 48$, $v_{\rm{move}} = 15$ km$/$h. The beam domain sCSI dimensions are  $N_{\upvartheta} = 2A$, $N_{\uptau} = 2K$ and $N_{\upnu} = 4{N_{\rm p}}$.}
						\label{CCcompare}
					\end{figure}

					\section{Conclusions}\label{section6}
					In this paper, we proposed the channel charting with PCFs and demonstrated its effectiveness in sCSI acquisition. We firstly introduced the definition of PCF with a cluster-based GBSM, which compactly captures the physical characteristics of the channel through a set of parameters. We then proposed a methodology for PCF acquisition in massive MIMO-OFDM systems. By aggregating the acquired PCFs across different locations, we constructed a structured, location-indexed channel charting. Further, we proposed an efficient algorithm to acquire beam domain sCSI from the PCFs in the channel charting. Simulation results showed that the PCF-derived sCSI achieves accuracy comparable to conventional online probing methods while reducing overhead and ensuring channel estimation accuracy. Overall, the proposed channel charting with PCFs offers an environment-aware and scalable solution for  wireless communication systems. Future work will further investigate optimal charting granularity guided by physical channel models to balance the trade-off between acquisition performance and charting maintenance cost. Moreover, the proposed PCF can be naturally integrated into wireless digital twin approaches as virtual digital objects for flexible sCSI reconstruction. Integration with generative digital twin models can further facilitate finer-resolution and more accurate channel charting.

					\appendices
					\section{Proof of Theorem \ref{theorem1}}\label{app1}
					According to the definition of tensor n-mode product, elements in ${\mathbfcal A}\left( {\mathbfcal Q} \right)$ can be calculated as
					\begin{align}
						&\scalebox{0.9}{${[{\mathbfcal A}\left( {\mathbfcal Q} \right)]_{{n_\upvartheta },{n_\uptau },{n_\upnu }}} =$} \scalebox{0.98}{$\sum\limits_{{n_\upvartheta }' = 1}^{{N_\upvartheta }} {\sum\limits_{{n_\uptau }' = 1}^{{N_\uptau } } {\sum\limits_{{n_\upnu }' = 1}^{{N_\upnu }} {{{[{{\bf{T}}^{\rm{s}}}]}_{{n_\upvartheta },{n_\upvartheta }'}}{{[{{\bf{T}}^{\rm{f}}}]}_{{n_\uptau },{n_\uptau }'}}} } }$} \notag\\
						&\scalebox{0.9}{$\quad\quad\quad\quad\quad\quad\quad\quad{\left[ {{{\bf{T}}^{\rm{t}}}} \right]_{{n_\upnu },{n_\upnu }'}}{[({\mathbfcal Q} \odot {\mathbfcal Q})]_{{n_\upvartheta }',{n_\uptau }',{n_\upnu }'}}. $}
					\end{align}
					With the problem in (\ref{finalp}) and the definition of ${L\left( {\mathbfcal Q} \right)}$ in (\ref{LQ}), we have (\ref{der}) on the top of the next page. Then the following equation holds:
					\begin{align}\label{holds}
						&	\scalebox{1.1}{$\frac{{\partial {{[{\mathbfcal A}\left( {\mathbfcal Q} \right)]}_{{n_\upvartheta },{n_\uptau },{n_\upnu }}}}}{{\partial {{\left[ {\mathbfcal Q} \right]}_{{n_{\upvartheta 1}},{n_\uptau }_1,{n_{{\upnu _1}}}}}}} =$} \notag\\
						&	2{[{{\bf{T}}^{\rm{s}}}]_{{n_\upvartheta },{n_{\upvartheta 1}}}}{[{{\bf{T}}^{\rm{f}}}]_{{n_\uptau },{n_{\uptau _1}}}}{\left[ {{{\bf{T}}^{\rm{t}}}} \right]_{{n_\upnu },{n_{\upnu_ 1}}}}{\left[ {\mathbfcal Q} \right]_{{n_{\upvartheta _1}},{n_\uptau }_1,{n_{{\upnu _1}}}}}. 
					\end{align}
					
					By substituting (\ref{holds}) into (\ref{der}), (\ref{final2}) on the top of the next page can be derived. Since ${{{\bf{T}}^{\rm{s}}}}$, ${{{\bf{T}}^{\rm{f}}}}$  and ${{{\bf{T}}^{\rm{t}}}}$ are symmetric matrices, and by applying the computation rules of the $n$-mode product of tensors, (\ref{final2}) can be reformulated into the expression  in (\ref{result}).
					
					\begin{figure*}
						
						\begin{align}\label{der}
							\scalebox{1}{$\frac{{\partial L\left( {\mathbfcal Q} \right)}}{{\partial {{\left[ {\mathbfcal Q} \right]}_{{n_{\upvartheta 1}},{n_\uptau }_1,{n_{{\upnu _1}}}}}}} = \frac{{\sum\limits_{{n_\upvartheta } = 1}^{{N_\upvartheta } } {\sum\limits_{{n_\uptau } = 1}^{{N_\uptau } } {\sum\limits_{{n_\upnu } = 1}^{{N_\upnu } } {\partial {{[{\mathbfcal A}\left( {\mathbfcal Q} \right)]}_{{n_\upvartheta },{n_\uptau },{n_\upnu }}}} } } }}{{\partial {{\left[ {\mathbfcal Q} \right]}_{{n_{\upvartheta 1}},{n_\uptau }_1,{n_{{\upnu _1}}}}}}} -\!\!\!\! \sum\limits_{{n_\upvartheta } = 1}^{{N_\upvartheta } } {\sum\limits_{{n_\uptau } = 1}^{{N_\uptau } } {\sum\limits_{{n_\upnu } = 1}^{{N_\upnu } } \!\!\!{\frac{{{{[\bar {\mathbfcal A}]}_{{n_\upvartheta },{n_\uptau },{n_\upnu }}}}}{{{{[{\mathbfcal A}\left( {\mathbfcal Q} \right) ]}_{{n_\upvartheta },{n_\uptau },{n_\upnu }}}}}\frac{{\partial {{[{\mathbfcal A}\left( {\mathbfcal Q} \right)]}_{{n_\upvartheta },{n_\uptau },{n_\upnu }}}}}{{\partial {{\left[ {\mathbfcal Q} \right]}_{{n_{\upvartheta 1}},{n_\uptau }_1,{n_{{\upnu _1}}}}}}}} } }.$}
						\end{align}

						\begin{align} \label{final2}
							\scalebox{1}{$\frac{{\partial L\left( {\mathbfcal Q} \right)}}{{\partial {{\left[ {\mathbfcal Q} \right]}_{{n_{\upvartheta _1}},{n_\uptau }_1,{n_{{\upnu _1}}}}}}} = 2{\left[ {\mathbfcal Q} \right]_{{n_{\upvartheta _1}},{n_\uptau }_1,{n_{{\upnu _1}}}}}\sum\limits_{{n_\upvartheta } = 1}^{{N_\upvartheta }} {\sum\limits_{{n_\uptau } = 1}^{{N_\uptau } } {\sum\limits_{{n_\upnu } = 1}^{{N_\upnu } } {{{[{{\bf{T}}^{\rm{s}}}]}_{{n_\upvartheta },{n_{\upvartheta 1}}}}{{[{{\bf{T}}^{\rm{f}}}]}_{{n_\uptau },{n_{\uptau_1}}}}{{\left[ {{{\bf{T}}^{\rm{t}}}} \right]}_{{n_\upnu },{n_{\upnu _1}}}}(1 - \frac{{{{[\bar {\mathbfcal A}]}_{{n_\upvartheta },{n_\uptau },{n_\upnu }}}}}{{{{[{\mathbfcal A}\left( {\mathbfcal Q} \right)]}_{{n_\upvartheta },{n_\uptau },{n_\upnu }}}}})} } }. $}
						\end{align}
						\hrulefill
					\end{figure*}

					\section{Proof of Theorem \ref{theorem2}}\label{app2}
					We take ${{\bf{T}}^{\rm{s}}}$ as an example to prove the validity of (\ref{T21}); the proofs of (\ref{T22}) and (\ref{T23}) can be derived in a similar manner. 	According to \cite{10146318}, for any circulant matrix ${\bf{C}} \in {{\mathbb{C}}^{N \times N}}$, it can be decomposed as ${\bf{C}} =$\scalebox{0.9}{$ {\bf{F}}_N^{ - 1}{\rm{Diag}}\left( {{{\bf{F}}_N}{{\left[ {\bf{C}} \right]}_{:,1}}} \right){{\bf{F}}_N}$}. Since the matrix ${\left( {{{\bf{V}}^{\rm{s}}}} \right)^{\rm{H}}}{{\bf{V}}^{\rm{s}}}$ is a circulant matrix, it can be written as 
					\begin{align}
						\!\!\!\!\!	{\left( {{{\bf{V}}^{\rm{s}}}} \right)^{\rm{H}}}{{\bf{V}}^{\rm{s}}}\! = {\bf{F}}_{{N_\upvartheta }}^{\rm{H}}{\rm{Diag}}({{\bf{p}}^{\rm{s}}}){{\bf{F}}_{{N_\upvartheta }}}\!=\! {{{N_\upvartheta }}}{\bf{F}}_{{N_\upvartheta }}^{{\rm{ - 1}}}{\rm{Diag}}({{\bf{p}}^{\rm{s}}}){{\bf{F}}_{{N_\upvartheta }}},
					\end{align}
					where \scalebox{0.9}{${{\bf{F}}_{{N_\upvartheta }}}{\left[ {{{\left( {{{\bf{V}}^{\rm{s}}}} \right)}^{\rm{H}}}{{\bf{V}}^{\rm{s}}}} \right]_{:,1}}$}$ = {N_\upvartheta }{{\bf{p}}^{\rm{s}}}$ holds.

					After the element-wise multiplication, it can be observed that ${\mathbf{T}}^{\rm{s}}$ is also a circulant matrix. With the above property of circulant matrix, we have 
					\begin{equation}\label{Ts1}
						{{\bf{T}}^{\rm{s}}} = {\bf{F}}_{{N_\upvartheta }}^{{\rm{ - 1}}}{\rm{Diag}}\left\{ {{{\bf{F}}_{{N_\upvartheta }}}{{\left[ {{{\bf{T}}^{\rm{s}}}} \right]}_{:,1}}} \right\}{{\bf{F}}_{{N_\upvartheta }}}.
					\end{equation}
					With ${\left[ {{{\bf{T}}^{\rm{s}}}} \right]_{:,1}} = $ \scalebox{0.95}{${\left[ {{{\left( {{{\bf{V}}^{\rm{s}}}} \right)}^{\rm{H}}}{{\bf{V}}^{\rm{s}}}} \right]_{:,1}} \odot \left[ {{{\left( {{{\bf{V}}^{\rm{s}}}} \right)}^{\rm{H}}}{{\bf{V}}^{\rm{s}}}} \right]_{:,1}^*$}, (\ref{Ts1}) can be further expressed as
					\begin{align}
						{{\bf{T}}^{\rm{s}}} &= {\bf{F}}_{{N_\upvartheta }}^{{\rm{ - 1}}}{\rm{Diag}}\left\{ {{{\bf{F}}_{{N_\upvartheta }}}\left( {{\bf{F}}_{{N_\upvartheta }}^{\rm{H}}{{\bf{p}}^{\rm{s}}} \odot {{\left( {{\bf{F}}_{{N_\upvartheta }}^{\rm{H}}{{\bf{p}}^{\rm{s}}}} \right)}^*}} \right)} \right\}{{\bf{F}}_{{N_\upvartheta }}}\notag\\
						&=\frac{1}{{{N_\upvartheta }}}{\bf{F}}_{{N_\upvartheta }}^{\rm{H}}{\rm{Diag}}\left\{ {{{\bf{F}}_{{N_\upvartheta }}}\left( {{\bf{F}}_{{N_\upvartheta }}^{\rm{H}}{{\bf{p}}^{\rm{s}}} \odot {{\left( {{\bf{F}}_{{N_\upvartheta }}^{\rm{H}}{{\bf{p}}^{\rm{s}}}} \right)}^*}} \right)} \right\}{{\bf{F}}_{{N_\upvartheta }}}\notag\\
						&= {\bf{F}}_{{N_\upvartheta }}^{\rm{H}}{\rm{Diag}}\{ {{\bf{t}}^{\rm{s}}}\} {{\bf{F}}_{{N_\upvartheta }}}.
					\end{align}
					This completes the proof of (\ref{T21}).

					\section*{Acknowledgment}
					The authors would like to thank the editor and the
					anonymous reviewers for their constructive comments and
					suggestions.

					\bibliographystyle{IEEEtran}
					\bibliography{Refabrv_20180823,FINAL_VERSION}

				\end{document}